\documentclass[12pt,reqno]{amsart}

\usepackage[active]{srcltx}

\usepackage{latexsym}
\usepackage{graphicx}
\usepackage{amsmath}
\usepackage{mathabx}
\usepackage{hyperref} 
\hypersetup{
    colorlinks,
    citecolor=black,
    filecolor=black,
    linkcolor=black,
    urlcolor=black
}
\newcommand{\Sch}{Schr\"odinger }
\newcommand{\set}[1]{\left\{#1\right\}}
\usepackage{color}

\def\Xint#1{\mathchoice
{\XXint\displaystyle\textstyle{#1}}%
{\XXint\textstyle\scriptstyle{#1}}%
{\XXint\scriptstyle\scriptscriptstyle{#1}}%
{\XXint\scriptscriptstyle\scriptscriptstyle{#1}}%
\!\int}
\def\XXint#1#2#3{{\setbox0=\hbox{$#1{#2#3}{\int}$ }
\vcenter{\hbox{$#2#3$ }}\kern-.6\wd0}}

\def\dashint{\Xint-}

\newcommand{\barg}{\mathcal{B}}

\newcommand{\edit}[1]{{\color{red}{$\clubsuit$#1$\clubsuit$}}}

\newcommand{\Szego}{\szego}

\newcommand{\kk}{\left( \frac{k}{2\pi} \right)}

\newcommand{\h}{\hat} 
\newcommand{\bex}{\begin{example}}
\newcommand{\eex}{\end{example}}
\newcommand{\bcs}{\begin{cases}}
\newcommand{\ecs}{\end{cases}}

\newcommand{\baa}{\begin{align*}}
\newcommand{\eaa}{\end{align*}}
\newcommand{\bea}{\begin{eqnarray*} }
\newcommand{\eea}{\end{eqnarray*} }
\newcommand{\beq}{\begin{equation} }
\newcommand{\eeq}{\end{equation} }
\newcommand{\bp}{\begin{proposition}}
\newcommand{\ep}{\end{proposition}}
\newcommand{\bt}{\begin{theorem}}
\newcommand{\et}{\end{theorem}}
\newcommand{\bpf}{\begin{proof}}
\newcommand{\epf}{\end{proof}}
\newcommand{\bl}{\begin{lemma}}
\newcommand{\el}{\end{lemma}}
\newcommand{\bc}{\begin{cor}}
\newcommand{\ec}{\end{cor}}
\newcommand{\bd}{\begin{definition}}
\newcommand{\ed}{\end{definition}}

\newcommand{\ihbar}{{\frac{i}{h}}}

\newcommand{\twiddle}[1]{\ensuremath{\widetilde{#1}}}

\newcommand{\be}{\begin{equation} }
\newcommand{\ee}{\end{equation} }
\newcommand{\bee}{\begin{eqnarray} }
\newcommand{\eee}{\end{eqnarray} }

\newcommand{\gives}{\ensuremath{\rightarrow}}

\newcommand{\x}{\ensuremath{\times}}

\newcommand{\EE}[1]{\mathbb E}
\newcommand{\abs}[1]{\left\lvert #1 \right\rvert}
\newcommand{\norm}[1]{\left\lVert#1\right\rVert}
\newcommand{\lr}[1]{\ensuremath{\left(#1\right)}}
\renewcommand{\Re}{\ensuremath{\mathrm{Re} \ }}
\renewcommand{\Im}{\ensuremath{\mathrm{Im} \ }}
\newcommand{\dell}{\ensuremath{\partial}}

\DeclareMathOperator{\sgn}{sgn}

\renewcommand{\Re}{{\operatorname{Re}\,}}
\renewcommand{\Im}{{\operatorname{Im}\,}}

\renewcommand{\epsilon}{\varepsilon}

\newcommand{\szego}{Szeg\H{o} }

\newcommand{\wt}{\widetilde}

\newcommand{\PP}{{\mathbb P}}

\newcommand{\R}{{\mathbb R}}
\newcommand{\C}{{\mathbb C}}

\newcommand{\Z}{{\mathbb Z}}

\newcommand{\CP}{\C\PP}

\newcommand{\half}{{\textstyle \frac 12}}

\newcommand{\supp}{{\operatorname{Supp\,}}}
\renewcommand{\phi}{\varphi}

\newcommand{\acal}{\mathcal{A}}
\newcommand{\bcal}{\mathcal{B}}
\newcommand{\ccal}{\mathcal{C}}

\newcommand{\hcal}{\mathcal{H}}

\newcommand{\lcal}{\mathcal{L}}

\newcommand{\ncal}{\mathcal{N}}

\newcommand{\scal}{\mathcal{S}}

\newcommand{\ucal}{\mathcal{U}}

\newcommand{\wcal}{\mathcal{W}}

\DeclareMathOperator{\Ai}{Ai}

\newtheorem{theo}{{\sc Theorem}}[section]

\newtheorem{prop}[theo]{{\sc Proposition}}
\newtheorem{defn}[theo]{{\sc Definition}}

\newtheorem{cor}[theo]{{\sc Corollary}}

\newtheorem{Lem}[theo]{{\sc Lemma}}

\newtheorem{proposition}[theo]{{\sc Proposition}}
\newtheorem{lemma}[theo]{{\sc Lemma}}

\newenvironment{example}{\medskip\noindent{\it Example:\/} }{\medskip}

\title[Interface asymptotics of Eigenspace Wigner distributions]
{Interface asymptotics of Eigenspace Wigner distributions for the
  Harmonic Oscillator }

\begin{document}
\author{Boris Hanin and Steve Zelditch}
\address{Department of Mathematics, Texas A\&M and Facebook AI Research, NYC}

\email[B. Hanin]{bhanin@tamu.edu}

\address{Department of Mathematics, Northwestern  University, Evanston, IL
60208, USA}
\email[S. Zelditch]{zelditch@math.northwestern.edu}

\thanks{Research partially supported by    NSF grant DMS-1810747.}
\maketitle

\begin{abstract}  Eigenspaces of the quantum isotropic Harmonic Oscillator 
  $\widehat{H}_{\hbar} : = - \frac{\hbar^2}{2} \Delta + \half ||x||^2$ on $\R^d$ have extremally high multiplicites and the eigenspace projections $\Pi_{\hbar, E_N(\hbar)} $ have special asymptotic properties. This article gives a detailed study of their Wigner distributions $W_{\hbar, E_N(\hbar)}(x, \xi)$. Heuristically, if $E_N(\hbar) = E$,  $W_{\hbar, E_N(\hbar)}(x, \xi)$ is the 
  `quantization' of the energy surface $\Sigma_E$,  and should be like the delta-function $\delta_{\Sigma_E}$  on  $\Sigma_E$;  rigorously, $W_{\hbar, E_N(\hbar)}(x, \xi)$ tends in a weak* sense to $\delta_{\Sigma_E}$. But its pointwise asymptotics and scaling asymptotics have more structure. The main results  give  Bessel  asymptotics of $W_{\hbar, E_N(\hbar)}(x, \xi)$ in the interior $H(x, \xi) < E$ of $\Sigma_E$; 
  {\it interface Airy scaling asymptotics} 
  in tubes of radius
  $\hbar^{2/3}$ around $\Sigma_E$, with 
  $(x, \xi)$ either in the interior or exterior of the energy ball; and exponential decay rate sin the exterior
  of the energy surface. 
\end{abstract}

\date{\today}

\section{Introduction} 
This article is part of a series \cite{HZZ15,ZZ16} studying the scaling asymptotics of spectral projection kernels along interfaces between allowed and forbidden regions. In this article, we are concerned with semi-classical Wigner distributions,
\begin{equation}\label{WIGNERDEF1}
 W_{\hbar, E_N(\hbar)}(x, \xi) := \int_{\R^d} \Pi_{\hbar, E_N(\hbar)} \left( x+\frac{v}{2}, x-\frac{v}{2} \right) e^{-\frac{i}{\hbar} v \cdot \xi} \frac{dv}{(2\pi h)^d} 
\end{equation} 
 of the eigenspace projections $\Pi_{\hbar, E_N(\hbar)} (x,y)$ for the isotropic Harmonic Oscillator
\begin{equation} \label{Hh}
\widehat{H}_{\hbar} =  \sum_{j = 1}^d \left(- \frac{\hbar^2}{2}   \frac{\partial^2 }{\partial
x_j^2} + \frac{  x_j^2}{2} \right) \quad \mathrm{on}\quad L^2(\R^d,dx).
\end{equation}
The notation in \eqref{WIGNERDEF1} is as follows: the  spectrum of \eqref{Hh} consists of the eigenvalues 
\begin{equation} \label{ENh} E_N(\hbar)=\hbar\lr{N+d/2},\qquad  \;\;  (N = 0, 1, 2, \dots)\end{equation} 
with multiplicities given by the composition function $p(N,d)$ of $N$ and $d$ (i.e. the number of ways to write $N$ as an ordered sum of $d$ non-negative integers). That is,
the eigenspaces 
\begin{equation} \label{VhE} V_{\hbar, E_N(\hbar)}: = \{\psi \in L^2(\R^d): \widehat{H}_{\hbar} \psi = 
E_N(\hbar) \psi \}, \end{equation}
have dimensions given by 
\begin{equation} \label{MULTS} \dim V_{\hbar, E_N(\hbar)} = p(N,d) \simeq N^{d-1}. \end{equation}
Due to extreme degeneracy of the spectrum of \eqref{Hh} when $d\geq 2$, the eigenspace projections 
\begin{equation} \label{PiDEF} 
\Pi_{\hbar, E_N(\hbar)}: L^2(\R^d) \to V_{\hbar,E_N(\hbar)}
\end{equation}
have very special properties reflecting the periodicity of the classical Hamiltonian flow and of the Schr\"odinger propagator $\exp[- \frac{it}{\hbar} \widehat{H}_{\hbar}]$  (see Definition \ref{WIGNERPROJDEF} for notation).  The focus of this article is on the Wigner distributions of individual eigenspace projections.

\begin{defn} \label{WIGNERPROJDEF} The Wigner distributions $W_{\hbar, E_N(\hbar)}(x, p) \in L^2(T^*\R^d)$ of
the  eigenspace  projections $\Pi_{\hbar, E_N(\hbar)}$ are defined by,
\begin{equation}
 W_{\hbar,E_N(\hbar)}(x, \xi) = \int_{\R^d} \Pi_{\hbar, E_N(\hbar)} \left( x+\frac{v}{2}, x-\frac{v}{2} \right) e^{-\frac{i}{\hbar} v \cdot \xi} \frac{dv}{(2\pi h)^d} \label{E:Wignera}.
\end{equation}  
\end{defn}

\noindent When $E_N(\hbar)  = E$, i.e. 
\begin{equation} \label{hNEDEF}\hbar = \hbar_N(E):=\frac{E}{N+\frac{d}{2}},\end{equation} the Wigner distribution $W_{\hbar, E_N(\hbar)}$ of a single eigenspace projection \eqref{E:Wignera} is the `quantization' of the energy surface of energy
$E$ and should therefore be localized at the classical energy level $H(x, \xi)  = E$, where  $H(x, \xi) = \half \sum_{j=1}^d (\xi_j^2 +  x_j^2) $. We denote the (energy) level sets by
\begin{equation} \label{SIGMAEDEF}
 \Sigma_E  =\{(x, \xi) \in T^*\R^d: H(x, \xi): = \half(||x||^2 + ||\xi||^2) = E\}. \end{equation}
The Hamiltonian flow of $H$ is $2 \pi $ periodic, and its orbits form the
complex projective space $\CP^{d-1} \simeq \Sigma_E /\sim$ where $\sim$ is the equivalence relation of belonging to the same Hamilton orbit. Due to this periodicity, the projections \eqref{PiDEF} are semi-classical Fourier integral operators (see \cite{GU12, HZZ15}). This is also true for the Wigner distributions \eqref{E:Wignera}.  Their properties are basically unique to the isotropic oscillator \eqref{Hh} and that is the motivation to single them out in this article\footnote{See  Section \ref{PRIOR} for a more precise statement.}. These properties are visible in Figure \ref{fig-Wigner-eigenspace} depicting the graph of $W_{\hbar, 1/2}$.

\begin{figure}\label{fig-Wigner-eigenspace}
\begin{center}
  \includegraphics[width=.6 \textwidth]{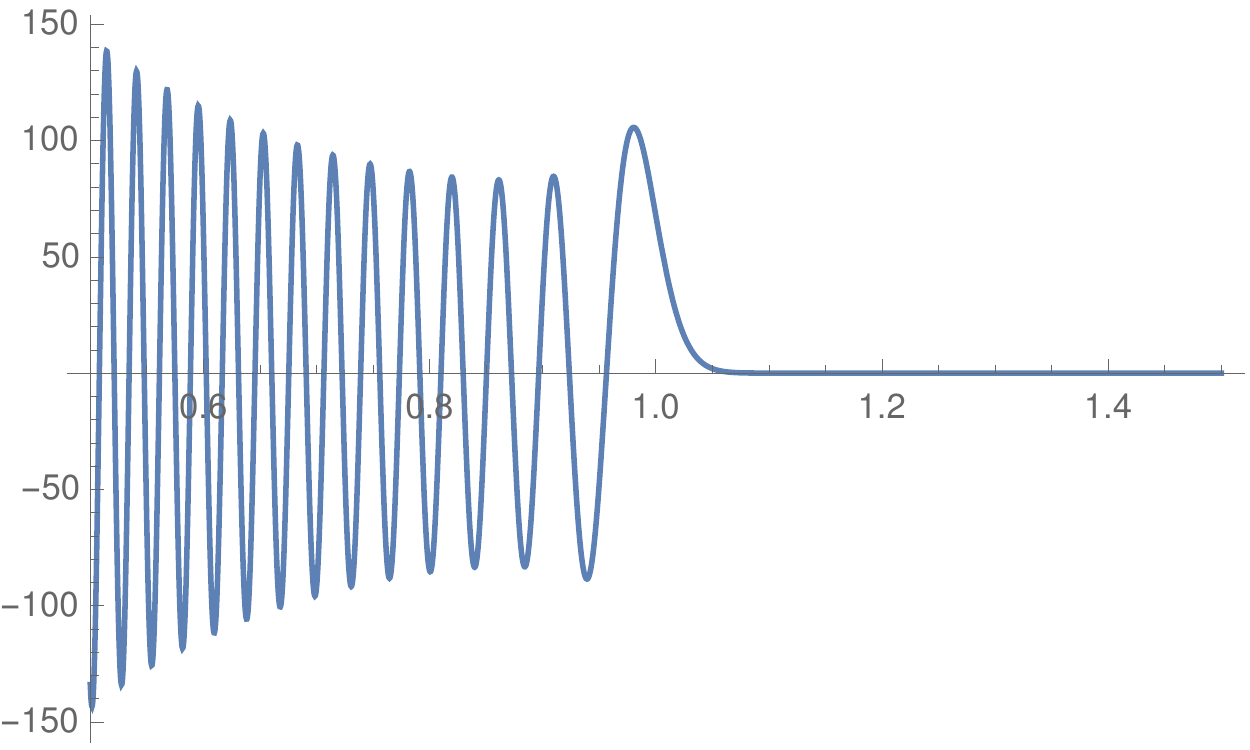} 
\end{center}
\caption{The Wigner function $W_{\hbar, E_N(\hbar)}$ of the eigenspace projection $\Pi_{\hbar, E_N(\hbar)}$ is always radial (see Proposition \ref{WIGNERLAGUERRE}). Displayed above is the graph of $W_{\hbar, E_N(\hbar)}$ as a function of the radial variable when $N,\hbar$ are such that $E_N(\hbar)=\hbar(N+d/2)=1/2$.}
\vspace{-10pt}
\end{figure}

Wigner distributions were introduced in \cite{W32} as phase space densities. Heuristically, the Wigner distribution \eqref{PiDEF}  is a kind of probability density in phase space of finding a particle of energy $E_N(\hbar)$ at the
point $(x, \xi) \in T^* \R^d$. This is not literally true, since $W_{\hbar, E_N(\hbar)}(x, \xi)$ is not positive: it oscillates with  heavy tails inside the energy surface \eqref{SIGMAEDEF}, 
has a kind of transition across $\Sigma_E$ and then
decays  rapidly outside the energy surface.  The purpose of this paper is to give detailed results on  the concentration and oscillation properties of these Wigner distributions in three phase space regimes, depending on
the position of $(x, \xi)$  with respect to $\Sigma_E$. 
The main results of this article are: \begin{enumerate}

\item  Theorem \ref{BESSELPROP}, giving the interior Bessel  asymptotics for $H(x, \xi) < E$; \bigskip

\item Theorem \ref{SCALINGCOR-old}, 
giving the
{\it interface Airy scaling asymptotics} of the Wigner distributions  \eqref{E:Wignera}
  in tubes of radius
$\hbar^{2/3}$ around $\Sigma_E$, with 
 $(x, \xi)$ either in the interior or exterior of the energy ball; \bigskip
 
  \item  Proposition \ref{EXTDECAY}, giving the exponential decay rate in the exterior
  of the energy surface. \bigskip\end{enumerate}
 
Of these results, the scaling asymptotics around $\Sigma_E$ is the most novel and is part of a more general program to analyse spectral scaling asymptotics around various types of interfaces. In a subsequent article \cite{HZ19A}, we consider Wigner distributions of more general spectral projections for `windows' of eigenvalues of varying widths and centers. When one sums over large enough windows, the scaling asymptotics should be universal, and that is the subject of \cite{HZ19B}.  Wigner distributions of individual eigenspace projections of the isotropic harmonic oscillator have special asymptotics and are not universal. For  general  \Sch operarators, whose classical Hamiltonian flows are non-periodic and whose
eigenspaces are not of `maximal multiplicity',  the generalization of  individual eigenspace projections of \eqref{Hh}  is  spectral projections for a window of width $\hbar$ around an energy level. In general, even for generic Harmonic oscillators, one would get an asymptotic expansion in terms of periodic orbits. Since all orbits of the classical isotropic oscillator are periodic, the asymptotics may be stated without reference to periodic orbits.

\subsection{Statement of results}

\begin{figure}\label{fig-Wigner-eigenspace}
\begin{center}
  \includegraphics[width=.6 \textwidth]{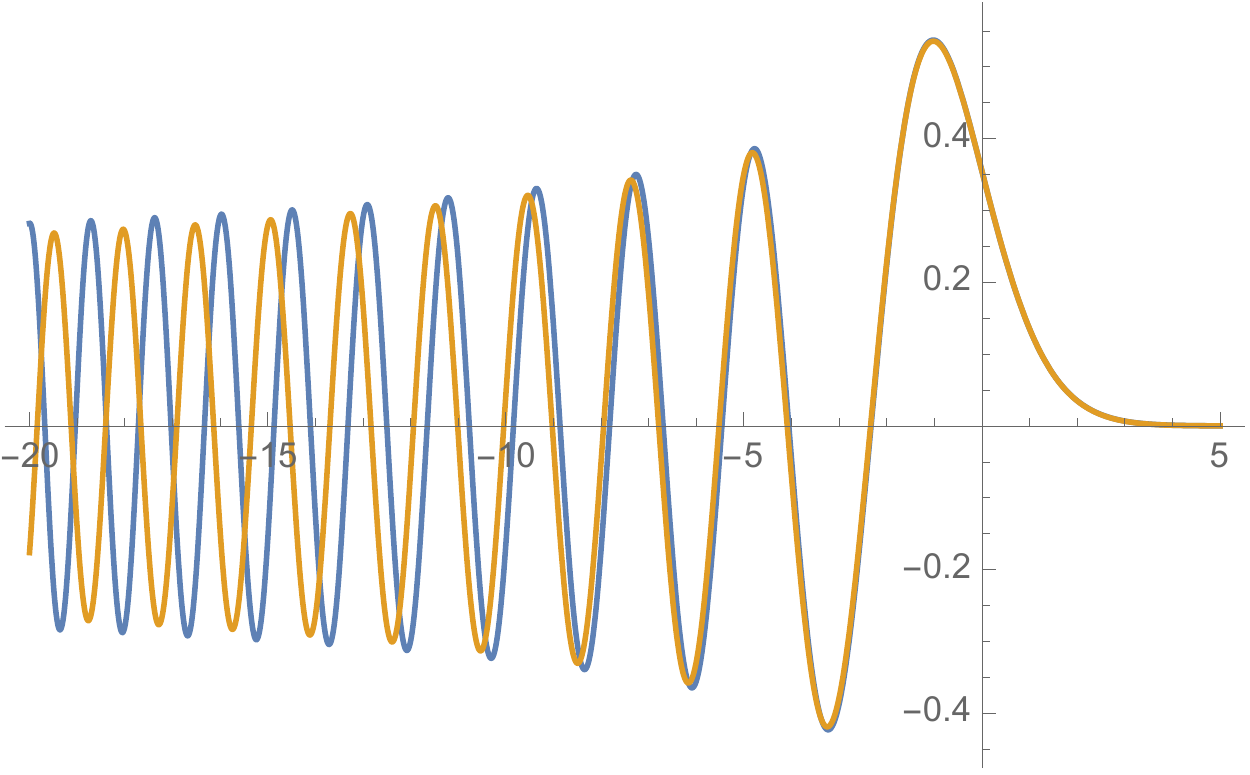} 
\end{center}
\caption{The Wigner function $W_{\hbar, E_N(\hbar)}$ of the eigenspace projection $\Pi_{\hbar, E_N(\hbar)}$ is always radial (see Proposition \ref{WIGNERLAGUERRE}). Displayed above is the graph of the Airy function (orange) and of $W_{\hbar, E_N(\hbar)}$ with $N=500$ (blue) as a function of the rescaled radial variable $\rho$ in a $\hbar^{2/3}$ tube around the energy surface $H(x,\xi)=E_N(\hbar)=1/2.$ Theorem \ref{SCALINGCOR-old} predicts that, when properly scaled, $W_{\hbar, E_N(\hbar)}$ should converge to the Airy function (with the rate of convergence being slower farther from the energy surface, which is defined here by $\rho=0$).}
\vspace{-10pt}
\end{figure}

%


 The first result is an exact formula for the Wigner distributions \eqref{E:Wignera} of the eigenspace projections for the isotropic Harmonic oscillator in terms of Laguerre functions (see Appendix \ref{S:Laguerre}
 and \cite{T} for background on Laguerre functions).

\begin{proposition} \label{WIGNERLAGUERRE} The Wigner distribution of 
Definition \ref{WIGNERPROJDEF} 
 is given by,
\begin{equation}    \label{E:Wigner-sp}
    W_{\hbar, E_N(\hbar)}(x, \xi) =  \frac{(-1)^N}{(\pi \hbar)^d}
    e^{-  2H/\hbar}  L^{(d-1)}_N(4H/\hbar),\qquad H=H(x,\xi)=\frac{\abs{x}^2+\abs{\xi}^2}{2},
\end{equation}
where $L_N^{(d-1)}$ is the associated Laguerre polynomial of degree $N$ and type $d-1$.
\end{proposition}

In dimension $d =1$ the formula was proved in \cite{O,JZ}. The authors did not find the formula in the literature for  $d > 1$, but essentially the same formula is proved in \cite[Theorem 1.3.5]{T} for matrix elements of the Heisenberg group; by \cite[(1.89)]{F} the latter are related to the Wigner distributions by the change of variables $(x, \xi) \to (-2\xi, 2 x)$ and a multiplication by $2^d$. We follow in Section \ref{S:Wigner-Laguerre-Pf} the technique in \cite{T} to give a brief derivation of Proposition \ref{WIGNERLAGUERRE} as a corollary of Proposition \ref{ucalprop-new} below. A proof that $W_{\hbar, E_N(\hbar)}$
is a radial function on $T^* \R^n$ using its symmetry properties and no special calculations is given in Section \ref{RADIALSECT}.

The second result is a weak* limit result for normalized Wigner distributions.

\begin{prop}  \label{OPWa} Let  $a_0$ be a semi-classical symbol of order zero and let $Op_h^w(a)$ be its Weyl quantization.  Then, as $\hbar \to 0$, with
$E_N(\hbar) \to E$,
$$\frac{1}{\dim V_{\hbar, E_N(\hbar)}} \int_{T^* \R^d} a_0(x,\xi)  W_{\hbar, E_N(\hbar)}(x,\xi) dx d \xi\quad \to\quad \dashint_{\Sigma_E} a_0 d \mu_E,$$
where $d\mu_E$ is Liouville measure on $\Sigma_E$ and $ \dashint_{\Sigma_E} a_0 d \mu_E = \frac{1}{\mu_E(\Sigma_E)}  \int_{\Sigma_E} a_0 d \mu_E$.
\end{prop}
 Thus,  $W_{\hbar, E_N(\hbar)}(x,\xi) \to \delta_{\Sigma_E}$ in the sense of weak* convergence. But this limit is due to the oscillations inside the energy ball; the  pointwise asymptotics are far more complicated. The proof is given in Section \ref{PROOFPROPOPWa} as a corollary of the pointwise asymptotics of the Wigner function (or, more precisely, its proof).

 \subsection{Interface asymptotics for Wigner distributions of  individual eigenspace projections}

Our first main result gives the scaling asymptotics for the Wigner function $W_{\hbar, E_N(\hbar)}(x,\xi)$ of the projection onto the $E$-eigenspace of $\widehat{H}_{\hbar}$ when $(x,\xi)$ lies in an $\hbar^{2/3}$ neighborhood of the energy surface $\Sigma_E.$
\begin{theo}\label{SCALINGCOR-old} 
Fix $E>0,d\geq 1,\epsilon>0$. Assume $E_N(\hbar)  = E$ and let $\hbar = \hbar_N(E),$ as in \eqref{hNEDEF}. Suppose $(x,\xi)\in T^*\R^d$ satisfies
\begin{equation} \label{udef}
H(x,\xi)=E + u \lr{\frac{\hbar}{2E}}^{2/3},\qquad u \in \R,\, H(x,\xi)=\frac{\norm{x}^2+\norm{\xi}^2}{2}
\end{equation}
with $\abs{u}<\hbar^{(-1+\epsilon)/3}$. Then, 
\begin{equation}
  \label{E:W-scaling}
 W_{\hbar, E_N(\hbar)}(x, \xi) = \frac{2}{(2\pi \hbar)^d} \lr{\frac{\hbar}{2E}}^{1/3} \lr{\Ai(u/E) ~+~ O(\hbar^{1-\epsilon}\abs{u}^{3})}.
\end{equation}
When $u>0,$ we also have the estimate
\[W_{\hbar, E_N(\hbar)}(x, \xi) =\frac{2}{(2\pi \hbar)^d} \lr{\frac{\hbar}{2E}}^{1/3} \Ai(u/E)\lr{1 + O\lr{(1+\abs{u})^{3/2}\abs{u}\hbar^{2/3}}}.\]
\end{theo}

\noindent Here, $\Ai(x)$ is the Airy function. The Airy scaling of $W_{\hbar, E_N(\hbar)}$ is illustrated in Figure 2. We give two proofs of Theorem \ref{SCALINGCOR-old}. The first is based on special functions; the asymptotics are obtained from Proposition  \ref{WIGNERLAGUERRE} and results on Laguerre functions due to Olver \cite{O}, Franzen-Wong \cite{FW} (see Proposition \ref{FW}). The second proof is self-contained and uses semi-classical parametrices  for $W_{\hbar,E_N(\hbar)}$ (see Section \ref{SCSCORSECT}).

The assumption \eqref{udef} may be stated more invariantly that $(x, \xi)$ lies in the tube of radius $O(\hbar^{2/3})$ around $\Sigma_E$ defined by the gradient flow of $H$ with respect to the Euclidean metric on $T^*\R^d$. The asymptotics are illustrated in Figure 2. Due to the behavior of the Airy function $\Ai(s)$, these formulae show that in the semi-classical limit $\hbar \to 0$, $E_N(\hbar) \to E$, $W_{\hbar, E_N(\hbar)}(x, \xi)$  concentrates on the energy surface surface $\Sigma_E$, is oscillatory inside the energy ball $\{H \leq E\}$ and is exponentially decaying outside the ball. There is a more complicated but more complete result which requires some notation $B(s)$ and $\alpha_0(s)$  defined in Section \ref{FW}; we defer the statement since it is too lengthy to define all the notation  here.

\subsection{\label{BESSELSECTINTRO} Interior Bessel and Trigonometric asymptotics}
In addition to the Airy asymptotics in an $\hbar^{2/3}$-tube around $\Sigma_E$, classical asymptotics of Laguerre functions (see \cite{AT15, AT15b} for recent results and references) show that when $H(x,\xi)\approx \hbar^2$ (i.e. $(x,\xi)$ is a dsitance at most $\hbar$ form the origin in phase space), $W_{\hbar, E_N(\hbar)}(x,\xi)$ exhibits Bessel asymptotics and trigonometric asymptotics farther into the interior of $\Sigma_E$ (i.e. when $H(x,\xi)\gg \hbar^2$). We record the precise asymptotic statements in Theorem \ref{BESSELPROP} below. To state it we define for $t \in [0, 1)$, 
$$A(t) = \half [\sqrt{t -t^2} + \sin^{-1} \sqrt{t}], \; t \in [0,1].$$
For $t <0$ the $\sin^{-1}$ is replaced by $\sinh^{-1}$ and the $\half$ by
$i/2$. The function $A(t)$ appears in the uniform asymptotic expansions of Laguerre polynomials $L_n^{(\alpha)}$ (see \cite[(2.7)]{FW}). Also, let $J_{d-1}$ be the Bessel function (of the first kind) of index $d-1$.

 \begin{theo} \label{BESSELPROP} Fix $E>0$ and suppose $E_N(\hbar)=E.$ For each $(x,\xi)\in T^*\R^d$ write
\[H_E := \frac{H(x,\xi)}{E}=\frac{\norm{x}^2+\norm{\xi}^2}{2E},\qquad \nu_E:=\frac{4E}{\hbar}.\]
Fix $0<a < 1/2.$ Uniformly over $a\leq H_E \leq 1-a$, there is an asymptotic expansion,
\begin{equation}
W_{\hbar, E_N(\hbar)}(x, \xi)=\frac{2}{(2\pi\hbar)^d}\left[\frac{J_{d-1}(\nu_E A(H_E))}{A(H_E)^{d-1}}\alpha_0(H_E)+O\lr{ \nu_E^{-1}\abs{\frac{J_{d}(\nu_E A(H_E))}{A(H_E)^{d}}}}\right].\label{E:Bessel-asymp}
\end{equation}
In particular, uniformly over $H_E$ in a compact subset of $(0,1),$ we find
\begin{equation}
W_{\hbar, E_N(\hbar)}(x, \xi) = (2\pi\hbar)^{-d+1/2} P_{H,E}\cos\lr{\xi_{\hbar, E,H}}+O\lr{\hbar^{-d+3/2}},\label{E:interior-cosine}
\end{equation}
where we've set
\[ \xi_{\hbar, E,H} =-\frac{\pi}{4} -\frac{2H}{\hbar}\lr{H_E^{-1}-1}^{1/2}+\frac{2E}{\hbar}\cos^{-1}\lr{H_E^{1/2}}\]
and 
\[P_{E,H}:=\lr{\pi E^{1/2}\lr{H_E^{-1}-1}^{1/4}\lr{H_E}^{d/2}}^{-1}.\]
\end{theo}

\noindent We prove Theorem \ref{BESSELPROP} in Section \ref{BESSELSECT}. We also give an independent semiclassical proof of \eqref{E:interior-cosine} in Section \ref{S:sc-cosine-pf}. Note that $A(t)\sim \sqrt{t}$ when $t$ is near $0.$ Hence, although Theorem \ref{BESSELPROP} does not strictly apply when $(x,\xi)$ is a distance $\hbar$ form the origin, the relation \eqref{E:Bessel-asymp} suggests that if the distance from $(x,\xi)$ to the origin is on the order $\hbar$, the Wigner function $W_{\hbar, E_N(\hbar)}(x, \xi)$ behaves like a normalized version of $J_{d-1}(t)/t^{d-1},$ up to constant factors (this fact can also be see by directly setting $H(x,\xi)=\hbar^2\rho$ in Proposition \ref{WIGNERLAGUERRE}). The small ball behavior of $W_{\hbar, E_N(\hbar)}(x,\xi)$ is investigated in Sections \ref{S:small-ball-ints} and \ref{S:sup} below (see also Figure 3).

\subsection{Small ball integrals }\label{S:small-ball-ints}

The interior Bessel asymptotics do not encompass the behavior of $W_{\hbar, E_N(\hbar)}$ in shrinking balls around $\rho = 0$. In that
case, we have,
\begin{prop} \label{SMBALLS} For $\epsilon>0$ sufficiently small and for
any $a(x, \xi) \in C_b(T^*\R^d)$,
   \begin{equation}
\int_{T^*\R^d} a(x, \xi) 
W_{\hbar, E_N(\hbar)}(x,\xi) \psi_{\epsilon,\hbar}(x,\xi)d x d \xi= O(\hbar^{\frac{1-d}{2}-2d\epsilon}\norm{a}_{L^\infty(B_0(\hbar^{1/2-\epsilon}))}).
\label{E:localized}
\end{equation}
where $\psi_{\epsilon,\hbar}$ is a smooth radial cut-off that is identically $1$ on the ball of radius $\hbar^{1/2-\epsilon}$ and is identically $0$ outside the ball of radius $2\hbar^{1/2-\epsilon}.$ 
\end{prop}

\subsection{Exterior asymptotics} If $E_N(\hbar) \to E$, then $W_{\hbar, E_N(\hbar)}(x, \xi)$ concentrates on $\Sigma_E$ and is exponentially decaying in the complement $H=H(x,\xi)> E$. The precise statement is,
\begin{prop}\label{EXTDECAY} Suppose that $H_E=H(x, \xi)/E>1$ and let $E_N(\hbar)  =E$. Then, there exists $C_1>0$ so that
\[|W_{\hbar, E_N(\hbar)}(x, \xi)| \leq C_1 \hbar^{-d + \frac{1}{2}} 
e^{- \frac{2 E}{\hbar}[\sqrt{H_E^2 -H_E} - \cosh^{-1}\sqrt{H_E} ]}. \]
Moreover, as $H(x,\xi) \to \infty$, there exists $C_2>0$ so that
$$|W_{\hbar, E_N(\hbar)}(x, \xi)| \leq C_2 \hbar^{-d + \frac{1}{2}} e^{- \frac{2H(x,\xi)}{\hbar}}. $$
\end{prop}
\noindent This result is proved in Section \ref{EXTDECAYSECT} as a simple consequence of classical asymptotics of Laguerre functions reviewed in Section \ref{Laguerre-Scaling-Proof}.

 \subsection{\label{SUPSECT} Supremum at $\rho =0$} \label{S:sup}
 
 The origin $\rho = 0$ is the point at which $W_{\hbar, E_N(\hbar)}$ has its global maximum (see Figure \ref{fig-Wigner-0}).  We have,
   
\begin{equation} \label{BUP} \begin{array}{lll} W_{\hbar, E_N(\hbar)}(0,0) & = &
 \frac{(-1)^N}{(\pi \hbar)^d}
      L^{d-1}_N(0) =  \frac{(-1)^N}{(\pi \hbar)^d}  \frac{\Gamma(N + d)}{\Gamma(N +1) \Gamma(d)}  \simeq \frac{(-1)^N}{\pi ^d}  C_d \hbar^{-d} N^{d-1}.\end{array} \end{equation}
The last statement follows from the explicit formula
$   L^{(d-1)}_N(0) =  \frac{\Gamma(N + d)}{\Gamma(N +1) \Gamma(d)} = \frac{(N + d-1)!}{N! (d-1)!}$ (see e.g. \cite[(1.1.39)]{T}).

The fact that $0$ is a local maximum can be seen from the eigenvalue
equation \eqref{WIGNEREIG3} below; since it is radial function, a negative
value of its Laplacian at $0$ is equivalent to its Hessian being negative definite. The fact that it is a global maximum follows from  Proposition \ref{WIGNERLAGUERRE} and the known properties of Laguerre polynomials.  We briefly discuss why the maximum is so large in Section \ref{ZERO1}.

\begin{figure}\label{fig-Wigner-0}
\begin{center}
  \includegraphics[width=.6 \textwidth]{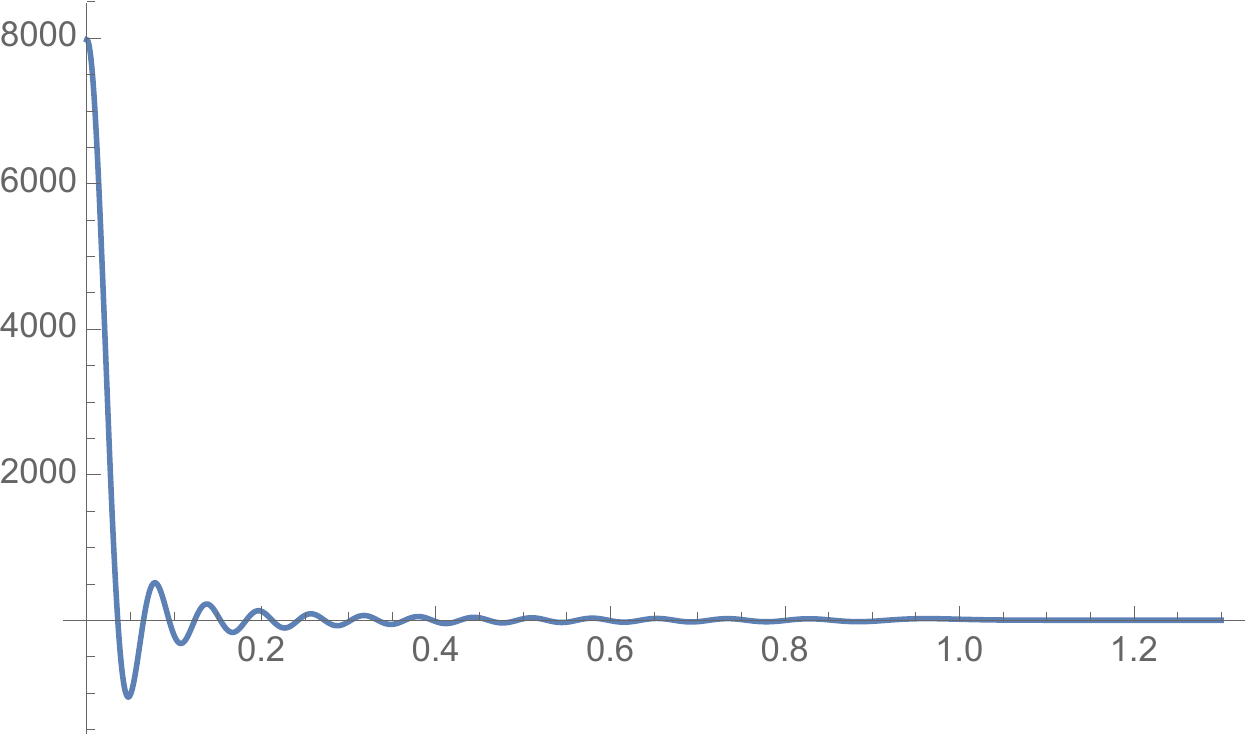} 
\end{center}
\caption{The Wigner function $W_{\hbar, E_N(\hbar)}$ of the eigenspace projection $\Pi_{\hbar, E_N(\hbar)}$ is always radial (see Proposition \ref{WIGNERLAGUERRE}). Displayed above is the blow-up of the Wigner function at $(0,0)$.}
\vspace{-10pt}
\end{figure}

On the complement of the ball  $B(0, \hbar^{\half - \epsilon})$, the Wigner
distribution is much smaller than at its maximum. The following is proved
by combining the estimates of Theorem \ref{SCALINGCOR-old} , Theorem \ref{BESSELPROP} and Proposition
\ref{EXTDECAY}.

\begin{prop} \label{WIGBOUND} For any $\epsilon > 0$,
$$\sup_{(x, \xi): H(x, \xi) \geq \epsilon} |W_{\hbar, E_N(\hbar)}(x, \xi)| \leq C \hbar^{-d + \frac{1}{3}}. $$
The supremum in this region  is achieved in $\{H \leq E\}$ at $(x,\xi)$
satisfying \eqref{udef} where $u$ is the global maximum of $\mathrm{Ai}(x)$.
\end{prop}

\subsection{\label{Outline}Outline of proofs } 

The main object in the  study of the eigenspace projections of Definition \ref{WIGNERPROJDEF} is the Wigner function of the propagator.  By `propagator' is meant the solution operator of the Cauchy problem for the  \Sch equation
$$i \hbar \frac{\partial}{\partial t} u = \widehat{H}_{\hbar} u,$$
and its Schwartz kernel is  thus given by
\begin{equation} \label{PROPDEF} 
U_h(t, x,y) =e^{-\ihbar t \widehat{H}_h}(x,y). \end{equation}
 The special feature of the isotropic oscillator is that we may express individual eigenspace projections as Fourier components of the Propagator $U_{\hbar}(t)= e^{- i \frac{t}{\hbar} \hat{H}_{\hbar}}$. On the level of Schwartz kernels,
 \begin{equation*}
   \Pi_{h, E_N(\hbar)}(x,y)=\int_{-\pi}^{\pi} U_\hbar(t-i\epsilon,x,y) e^{-\ihbar
(t-i\epsilon) E_N(\hbar)} \frac{dt}{2\pi},
 \end{equation*}
where we recall that $E_N(\hbar)=\hbar(N+d/2).$ The Wigner transform is linear on the level of Schwartz kernels and hence 
\begin{equation}
W_{\hbar, E_N(\hbar)}=\int_{-\pi}^{\pi} \ucal_\hbar(t-i\epsilon,x,y) e^{-\ihbar
(t-i\epsilon) E_N(\hbar)} \frac{dt}{2\pi},\label{PROJUT}
\end{equation}

where $\ucal_\hbar$ is the Wigner function of the propagator. 

\begin{proposition} \label{ucalprop-new} The Wigner distribution of $U_h(t, x, y)$ is given by $$
 \ucal_{\hbar}(t, x, \xi) = \frac{1}{(2\pi\hbar)^d}  \cos( t/2)^{-d}
 \exp\lr{-\frac{2iH}{\hbar} \frac{\sin (t/2)}{\cos (t/2)} },\qquad H = H(x,\xi):=\frac{||x||^2 + ||\xi||^2}{2}$$
 \end{proposition}

\noindent  Proposition \ref{ucalprop-new} is the basic tool underlying the scaling asymptotics of Wigner functions of eigenspace projections. The Wigner distribution is well-defined as a distribution (see Section \ref{WIGasDIST}) but not as a locally $L^1$ function. $\ucal_{\hbar}(t, x, \xi)$ is a distribution in $t$ for fixed $(x, \xi)$ with  Dirac mass singularities at $t = \pi + 4 \pi \Z$.  Despite the singularities it has scaling asymptotics in different phase space regimes. Combined with \eqref{PROJUT} it gives
scaling asymptotics of $W_{\hbar, E_N(\hbar)}(x, \xi). $

\subsection{\label{PRIOR}Prior and related results}

The literature on Wigner distributions and on the isotropic Harmonic Oscillator
is vast. In dimension one, some of the results of this article  are proved in \cite{JZ}, and others are stated and to some extent proved in 
\cite{Ber91}, along with more detailed asymptotics near interfaces.  Somewhat surprisingly, the results on Wigner distributions of spectral projections have not previously been generalized to higher dimensions $d \geq 2$, even for the model case \eqref{Hh} of the isotropic harmonic oscillator. However, there do exist known relations between Wigner functions
and Laguerre functions  \cite{T,F} and we explain how to use  known asymptotics of Laguerre
functions (from \cite{FW}) to obtain interface asymptotics  results on Wigner functions. But our ultimate aim is to generalize the results to more general \Sch operators, and therefore we also give semi-classical analysis arguments .

It has long been known in dimension $d  = 1$ \cite{Ber}  that Wigner distribution of spectral
projections (corresponding to individual eigenfunctions) of quite general \Sch operators exhibit Airy scaling asymptotics in $\hbar^{2/3}$ shells around the corresponding energy
 curve. In dimension $d =1$, Berry \cite{Ber} obtained the expression for the Wigner function
of an eigenfunction of any Hamiltonian, 
$$W_{\hbar}(q,p) = C \frac{A(q,p)^{1/6} \; \mathrm{Ai} \left(- \left[ \frac{3 A(q,p)}{2 \hbar} \right]^{2/3} \right)}{\pi \hbar^{\frac{2}{3}} \left[I_q(2) I_p(1) - I_p(2) I_1(q)\right]^{\half}}.$$
Here, $A(q,p)$ is the area between the chord $\overline{qp}$ and the arc
of the boundary between $q,p$. Of course, in  dimension one 
the \Sch operator is an ODE and the techniques available there do not
generalize  to higher dimensions.  See also \cite{O} for many results in the one-dimensional case.

     In \cite{HZZ16}, the authors studied the configuration space spectral
     projections kernels  $\Pi_{\hbar_N, E}$ for a fixed energy level $E$ 
for the isotropic Harmonic oscillator in $\R^d$.  The allowed region in
configuration space is  the projection $\acal_E = \pi(\Sigma_E) \subset \R^d$ where $\pi: T^*\R^d \to \R^d$ is the natural projection. 
The  interface is the
caustic $\ccal_E = \{x \in \R^d: \half ||x||^2 = E\}$, i.e. the projection 
of the points of $\Sigma_E$ where $\xi = 0$. Theorem 1.1 of \cite{HZZ16} gives Airy scaling asymptotics of   $\Pi_{\hbar_N, E}(x,y)$ for $x,y$ in an $\hbar^{2/3}$ neighborhood of $\ccal_E$.  Theorem \ref{SCALINGCOR-old} of the present article is a  lift of this result to phase space.

In \cite{ZZ16}, Zelditch-Zhou studied scaling asymptotics around $\Sigma_E$ of the
so-called Husimi distribution of $\Pi_{\hbar, E_N(\hbar)}$ rather than the Wigner distribution. 
The Husimi distribution is the covariant symbol (value on the diagonal) of the conjugate
$B_{\hbar}^* \Pi_{\hbar, E_N(\hbar)} B_{\hbar}^*$ of the eigenspace projection by the 
Bargmann transform to the holomorphic Bargmann-Fock space. The Husimi distribution 
is the density obtained by holomorphic quantization of the eigenspace projection and is a kind of Gaussian coherent state centered on $\Sigma$ with Gaussian decay in a tube 
of radius $N^{-\half}$ in the normal directions. Thus, the phase space interface scaling around $\Sigma_E$
is very different in the two representations. The exact relation between the two
phase space distributions was given by Cahill-Glauber \cite[(7.32)]{CG69I} (see also \cite[(6.32)] {CG69II} and \cite[(5.32)]{Oz}), who showed
that the Bargmann-Fock (Husimi distribution) is a Gaussian convolution of the Wigner distribution. 

At the beginning of this article, it is stated that the isotropic quantum oscillator \eqref{Hh}
is the unique \Sch operator with its exceptionally high eigenvalue multiplicities (degeneracies). Suppose that    $-h^2 \Delta + |x|^2 + V$ on $\R^d$ is  a \Sch operator
with a potential $V$ of quadratic growth, and suppose that  it has the same   eigenvalue multiplicities as the un-perturbed isotropic oscillator \eqref{Hh}.  Then it is a `maximally degenerate' \Sch operator in the sense of
\cite{Z96}.  As proved in that article, it must have periodic classical Hamiltonian 
flow and then, by results of  Weinstein, Widom and Guillemin on compact manifolds  (see \cite{G78} for background) and of Chazaraint for  \Sch operators on $\R^d$,   \cite{Ch80}, the eigenvalues concentrate in `clusters'  around the arithmetic progression $\hbar (N + \frac{d}{2})$;  a recent article studying perturbations of \eqref{Hh} is \cite{GU12}.  
Maximal degeneracy
means that every cluster has just one distinct eigenvalue. 
We are not aware of a proof
 that there exist no `maximally degenerate' perturbations of the semi-classial \Sch operator  \eqref{Hh} but it seems that the techniques of \cite{Z96} might yield the result that \eqref{Hh} is the unique maximally degenerate analytic  \Sch operator on $\R^d$ with a quadratic growth potential. We plan to consider this problem elsewhere and do not discuss it further here.

\section{Background}

It is well-known that the spectrum of $\widehat{H}_{\hbar}$ in the isotropic case consists of the eigenvalues \eqref{ENh}, and  one has the spectral decomposition 
\begin{equation}  
L^2(\R^d, dx) = \bigoplus_{N\in \mathbb N } V_{\hbar,E_N(\hbar)}, \;\; \widehat{H}_\hbar|_{V_{h,E_N(\hbar)}}  =  \hbar (N + \frac{d}{2})
\end{equation} 
An orthonormal basis of eigenfunctions of \eqref{Hh} is given by scaled Hermite functions,
\begin{equation}
\phi_{\alpha,h}(x)=h^{-d/4}p_{\alpha}\lr{x\cdot
h^{-1/2}}e^{-x^2/2h},\label{E:Scaling Relation}
\end{equation}
where $\alpha=\lr{\alpha_1,\ldots, \alpha_d}\geq (0,\ldots,0)$ is a
$d-$dimensional multi-index and
$p_{\alpha}(x)$ is the product $\prod_{j=1}^d p_{\alpha_j}(x_j)$ of the
hermite polynomials $p_k$ (of degree $k$)
in one variable. The eigenvalue of $\phi_{\alpha,h}$ is given by
\begin{equation} \label{EV}
H_{h} \phi_{\alpha,h} = h (|\alpha|+d/2) \phi_{\alpha,h}.
\end{equation}
The multiplicity of the eigenvalue $ h (|\alpha|+d/2)$ is the partition
function of $|\alpha|$, i.e. the number
of $\alpha=\lr{\alpha_1,\ldots, \alpha_d}\geq (0,\ldots,0)$  with a fixed
value of $|\alpha|$.  The spectral projections \eqref{PiDEF}  are given by
\begin{equation} \label{COV}
\Pi_{\hbar, E_N(\hbar)} (x, y) : = \sum_{|\alpha|=N} \phi_{\alpha,h_N}(x)
\phi_{\alpha,h_N}(y).   \end{equation}

 \subsection{Semi-classical scaling}

To clarify the roles of $E, N, \hbar$ and to  facilitate comparisons to references such as \cite{F} which do not employ the semi-classical scaling \eqref{COV}, we record the following notation throughout the article. 
\begin{itemize}
\item For \eqref{Hh} with $\hbar =1$, the eigenvalues are $E_N(1) = N + \frac{d}{2}$ with corresponding
eigenspaces $V_N = V_{1, E_N(1)}$ (see \eqref{ENh}).  We denote
by $\Pi_N: L^2(\R^d) \to V_N$ the orthogonal projections (again, with $\hbar = 1, E = E_N(1). $\bigskip

\item The projections $\Pi_{\hbar, E_N(\hbar)}$ \eqref{COV} (see also \eqref{ENh})  are related
to $\Pi_N$ by
$$\Pi_{\hbar, E_N(\hbar)}(x, y) = \hbar^{-d/2} \Pi_N(\hbar^{-1/2} x, \hbar^{-1/2} y). $$

\item The Wigner distribution of $\Pi_N$ is denoted by, 
\begin{equation} W_N(x, \xi) := 
\int_{\R^d} \Pi_{N} \left( x+\frac{v}{2}, x-\frac{v}{2} \right) e^{ i v p} \frac{dv}{(2\pi )^d} \label{E:Wignera1}.
\end{equation}

\item The Wigner distribution of $\Pi_{\hbar, E_N(\hbar)}$ is
\[W_{\hbar, E_N(\hbar)}(x,\xi)=\int_{\R^d} \Pi_{\hbar, E_N(\hbar)}\lr{x+\frac{v}{2}, x-\frac{v}{2}}e^{-\frac{i}{\hbar}v\cdot \xi}\frac{dv}{(2\pi \hbar)^d}.\]
It is related to that of $\Pi_N$ by
\begin{equation} \label{Halfscale} 
\begin{array}{l}  W_{\hbar,E_N(\hbar)}(x,\xi)=\hbar^{-d}W_{N}\lr{x/\hbar^{1/2},\, \xi/\hbar^{1/2}}
.\end{array} 
\end{equation} 
\end{itemize}


\subsection{Weyl pseudo-differential operators, metaplectic covariance and Wigner distributions}

A semi-classical Weyl pseudo-differential operator is defined by the formula, 
$$Op_h^w(a) u(x) = \int_{\R^d} \int_{\R^d} a_{\hbar}(\half(x + y), \xi) e^{\frac{i}{\hbar} \langle x - y, \xi \rangle } u(y) dy d \xi. $$ We refer to
\cite{F, Zw} for background.

A key property of Weyl quantization is the so-called metaplectic covariance. 
Let $Sp(2d, |R) =  Sp(T^* \R^d, \sigma)$ denote the symplectic group and let $\mu(g)$ denote
the metaplectic representation of its double cover. Then,
$\mu(g) Op_h^w(a) \mu(g)  = Op_h^w(a \circ T_g), $ where $T_g: T^*\R^d \to T^*\R^d$ denotes translation by $g$.

In particular, $U \in U(d)$ acts on $L^2(T^* \R^d)$ by translation $T_U$ of functions, using the identification $T^*\R^d \simeq \C^d$ defined by the standard complex structure $J$. 
 $U(d) \subset Sp(2d, \R) $ is a subgroup of the symplectic group and the complete symbol $H(x, \xi)$ of \eqref{Hh} is $U(d)$ invariant, so by metaplectic covariance, 
 $\hat{H}_{\hbar}$ commutes with the metaplectic represenation of $U(d).$

\subsection{\label{RADIALSECT} Proof that the Wigner distribution is radial}
The fact that the Wigner distribution \eqref{E:Wignera} is radial can be seen from the $U(d)$ symmetry of the isotropic oscillator without  calculations. 
As mentioned above, the  classical Hamiltonian  $H(x, \xi) = \half(||x||^2 + ||\xi||^2)$ 
is invariant under the standard action of the unitary group  $U(d)$ on $\C^d \simeq T^* \R^d$. Since $U(d) \subset Sp(2d, \R) = Sp(T^* \R^d, \sigma)$ is a subgroup of the symplectic group, the action is quantized by the metaplectic representation on  $L^2(\R^d)$  and commutes
with $\hat{H}_{\hbar}$ and therefore with $\Pi_{\hbar, E_N(\hbar)}$.
By metaplectic covariance, the Wigner function is invariant under the lift of the $U(n)$ action to $T^* \R^d$, so that,   all $U \in U(N)$, $$W_{U^*\Pi_{\hbar, E_N(\hbar)}U}(x,\xi) = W_{\hbar, E_N(\hbar)}(U(x, \xi)). $$
 This gives a simple explanation of the fact that  $W_{\hbar, E_N(\hbar)} $ is a function of $||x||^2 + ||\xi||^2$. In particular,
 the $S^1$ subgroup $e^{it \hat{H}}$ generated by the isotropic Harmonic Oscillator corresponds to the periodic
 Hamiltonian flow $\exp t \Xi_H$ of the classical oscillator, and $W_{\hbar, E-N(\hbar)}$ is invariant under this group.

 \subsection{\label{HSSECT}Trace and Hilbert-Schmidt properties}
 
 It follows that  $\Pi_{\hbar, E_N(\hbar)}$.
By using the identity
$$\langle Op^w(a) f, f \rangle = \int_{T^*\R^d} a(x, \xi) W_{f, f}(x, \xi) dx d\xi, $$ of \cite[Proposition 2.5]{F} for  orthonormal basis elements $f = \phi_{\alpha, \hbar_N}$ of $V_{\hbar, E_N(\hbar)}$ and summing over $\alpha$, one obtains the (well-known) identity, 
\begin{equation} \label{TRACEP} \mathrm{Tr} Op_h^w(a) \Pi_{\hbar, E_N(\hbar)} = \int_{T^* \R^d} a(x, \xi) W_{\hbar, E_N(\hbar)}(x, \xi) dx d\xi.  \end{equation}
Further, the Wigner transform \eqref{WIGNERDEF1} taking kernels to Wigner functions  is an isometry from Hilbert-Schmidt kernels $K(x,y)$ on $\R^d \times \R^d$ to their Wigner distributions
on $T^*\R^d$ \cite{F}.  From \eqref{TRACEP} and this isometry, it is straightforward to check that,
\begin{equation} \label{INTEGRALS}\left\{  \begin{array}{ll}(i) &  \int_{T^* \R^d} W_{\hbar, E_N(\hbar)}(x, \xi) dx d \xi = \mathrm{Tr} \Pi_{\hbar, E_N(\hbar)} = \dim V_{\hbar, E_N(\hbar)} = \binom{N+d-1}{d-1} \\ &  \\(ii) & 
\int_{T^* \R^d}\left| W_{\hbar, E_N(\hbar)}(x, \xi) \right|^2 dx d \xi = \mathrm{Tr} \Pi^2_{\hbar, E_N(\hbar)} = \dim V_{\hbar, E_N(\hbar)}=\binom{N+d-1}{d-1}\\ & \\ (iii)& \int_{T^* \R^d} W_{\hbar, E_N(\hbar)}(x, \xi) \overline{W_{\hbar, E_M(\hbar)}(x, \xi)}dx d \xi = \mathrm{Tr} \Pi_{\hbar, E_N(\hbar)}  \Pi_{\hbar, E_M(\hbar)}= 0, \; \mathrm{for}\; M \not=N. \end{array} \right.,\end{equation} 
In these equations, $N=\frac{E}{\hbar}-\frac{d}{2},$ and $\binom{N+d-1}{d-1}$ is the composition function of $(N,d)$ (i.e. the number of ways to write $N$ as an ordered us of $d$ non-negative integers). Thus, the sequence,
$$\{ \frac{1}{\sqrt{\dim V_{\hbar, E_N(\hbar)}}} W_{\hbar, E_N(\hbar)}\}_{N=1}^{\infty}\subset L^2(\R^{2n}) $$
is orthonormal.

In comparing \eqref{TRACEP}, \eqref{INTEGRALS}(i)-(ii) one should keep
in mind that $W_{\hbar, E_N(\hbar)}$ is rapidly oscillating in $\{H \leq E\}$
with slowly decaying tails in the interior of $\{H \leq E\}$, with a large `bump' near $\Sigma_E$  and with maximum
given by Proposition \ref{WIGBOUND}. Integrals (e.g. of $a \equiv 1$)
against $W_{\hbar, E_N(\hbar)}$ involve a lot of cancellation due to the
oscillations. The square integrals in (ii) enhance the `bump' and decrease the tails and of course are positive.

\subsection{The Wigner transform}

For any Schwartz kernel $K_{\hbar} \in L^2(\R^d \times \R^d)$ one may define the Wigner distribution of $K_{\hbar}$ by
 \begin{equation}
 W_{K, \hbar}(x, \xi): = (2 \pi\hbar)^{-d} \int_{\R^d} K_{\hbar} \left( x+\frac{v}{2}, x-\frac{v}{2} \right) e^{-\frac{i}{\hbar} v \xi} \frac{dv}{(2\pi h)^d}, \label{E:WignerK}
\end{equation} As mentioned in Section \ref{HSSECT},
the map from $K_\hbar \to W_{K, \hbar}$ defines a unitary operator $$\wcal_{\hbar}: L^2(\R^d \times \R^d) \to L^2(T^*\R^d),$$ 
which we will call the `Wigner transform'.

The unitary group $U(d)$ acts on $L^2(\R^d \times \R^d)$ by conjugation,$U(g) \cdot K = g K g^*$.
where we identify $K(x,y) \in L^2(\R^d \times \R^d)$ with the associated Hilbert-Schmidt operator.  Metaplectic covariance implies that,
$$\wcal_{\hbar} U(g) = T_g \wcal_{\hbar}. $$

\subsection{\label{EIGSECT} Eigenvalue equation}

We  employ the well-known  semi-classical Moyal $*_{\hbar}$ product on symbols of Weyl pseudo-differential
operators is defined by \cite{F, Sh}
$$Op_{\hbar}^w(a) \circ Op^w_{\hbar}(b) = Op^w_{\hbar}(a *_{\hbar} b). $$
The Wigner distribution $W_{\hbar, E_N(\hbar)}$ is the Weyl symbol of the eigenspace projection $\Pi_{\hbar, E_N(\hbar)}$.

The eigenvalue equation, 
$$\hat{H}_{\hbar} \Pi_{\hbar, E_N(\hbar)} = E_N(\hbar) \Pi_{\hbar, E_N(\hbar)} $$
 gives
\begin{equation} \label{WIGNEREIG} H *_{\hbar} W_{\hbar, E_N(\hbar)} = W_{\hbar, E_N(\hbar)} *_{\hbar} H =  E_N(\hbar) W_{\hbar, E_N(\hbar)}. \end{equation}

  It is proved in Section \ref{RADIALSECT} that $W_{\hbar, E_N(\hbar)}$ is a radial function on $\R^{2n}$, the explicit formula being
given in Proposition \ref{WIGNERLAGUERRE}. 
Up to a scalar fixed by the trace identities, $W_{\hbar, E_N(\hbar)}$ is the only radial solution of these equations.  Indeed,  
the equations  \eqref{WIGNEREIG}
are equivalent to,
\begin{equation} \label{WIGNEREIG2}\half  \left( \sum_j  \lr{x_j + \frac{i \hbar}{2} \partial_{\xi_j}}^2 +   \lr{\xi_j + \frac{i \hbar}{2} \partial_{x_j}}^2  \right) W_{\hbar,E_N(\hbar)}
= E_N(\hbar) W_{\hbar, E_N(\hbar)}.  \end{equation}
This equation is related to that of \cite[(1.3.14)]{T}.
The imaginary part vanishes,
$(x_j \partial_{\xi_j} - \xi_j \partial_{x_j} ) W_{\hbar, E} = 0, \;\; \forall j$, since
$W_{\hbar, E_N(\hbar)}$ is radial.
We then simplify  \eqref{WIGNEREIG2} to 
\begin{equation} \label{WIGNEREIG3} \left(- \frac{\hbar^2}{8}   (\Delta_{\xi} + \Delta_x) + H(x,\xi)\right) W_{\hbar,E_N(\hbar)}
= E_N(\hbar) W_{\hbar, E_N(\hbar)}.  \end{equation}
\subsection{\label{ZERO1}The special point $(x, \xi) = (0,0)$}
In this section we briefly return to the observation in Section \ref{SUPSECT}
that  $W_{\hbar, E_N(\hbar)}$ has its  global maximum at the origin
$(x, \xi) = (0,0)$ (see \eqref{BUP}).

As noted in Section \ref{EIGSECT}, $W_{\hbar, E_N(\hbar)}$ is an eigenfunction of an (essentially isotropic) \Sch operator \eqref{WIGNEREIG3} on $T^*\R^d$. 
 By
\cite[Lemma 10]{HZZ15},  the eigenspace spectral projections for the isotropic harmonic oscillator in dimension $d$ satisfies,
$$\Pi_{h,E}(x,x)=\lr{2\pi
  h}^{-(d-1)}\lr{2E-\abs{x}^2}^{\frac{d}{2}-1}\omega_{d-1}\lr{1+O(h)}, $$
  for a dimensional constant $\omega_d$. We apply this result to the eigenspace projections for \eqref{WIGNEREIG3} in dimension $2d$ and find that at the point $(0,0)$ its diagonal value is of order $\hbar^{-2d + 1}$.  We then express this eigenspace projection in terms of an orthonormal
  basis for the eigenspace. As discussed in Section \ref{HSSECT}, one orthonormalized term is given by $\frac{1}{\sqrt{\dim V_{\hbar, E_N(\hbar)}}} W_{\hbar, E_N(\hbar)}$. Note that  ${\dim V_{\hbar, E_N(\hbar)}} \simeq \hbar^{-2d + 1}$ in dimension $2d$. Due to the normalization and \eqref{BUP}, 
  $$\frac{1}{\sqrt{\dim V_{\hbar, E_N(\hbar)}}} W_{\hbar, E_N(\hbar)}(0,0) \simeq \hbar^{-2 d + 1 + d -\half} = \hbar^{-d + \half}. $$
  
  Pointwise Weyl laws for eigenfunctions of  \Sch operators show that this order of magnitude for the sup-norm of an $L^2$-normalized eigenfunction is
the maximum possible in the dimension $2 d$ of $T^* \R^d$. It is proved
in \cite{SoZ16} that when such `maximal eigenfunction growth' occurs, and if the \Sch is real-analytic, then
there must exist a point $q$ at which all Hamiltonian orbits with initial position
$q$ loop back to $q$ at a fixed time and that the first return map must preserve an $L^1$ function on $\Sigma_E(q)$. This is of course trivially true in the
case of \eqref{Hh}, since all orbits are periodic. So there does not exist a completely geometric explanation for the fact that the maximal order of growth is obtained at the origin.
  
But there exists a simple spectral geometric explanation for the order of magnitude at the origin: All eigenfunctions of \eqref{WIGNEREIG3} with the exception of the radial eigenfunction  $W_{\hbar, E_N(\hbar)}(0,0)$ vanish at the origin $(0,0)$ since they transform by non-trivial characters
  of $U(d)$ and $(0,0)$ is a fixed point of the action. Consequently, the value of the eigenspace projection on the diagonal at $(0,0)$ is the square of
   $W_{\hbar, E_N(\hbar)}(0,0)$ and that accounts precisely for the order of growth.
  
  
 We note some further interesting facts about the special point $(0,0)$.
 One has the exact formula, 
     
\begin{equation} \label{Q} \begin{array}{lll} W_{\hbar, E_N(\hbar)}(0,0) & = &
  \int_{\R^d} \Pi_{\hbar, E_N(\hbar)}(\frac{v}{2}, - \frac{v}{2}) dv = \mathrm{Tr}\; \tau \Pi_{\hbar, E_N(\hbar)}, \end{array} \end{equation}
where $\tau(x) = -x$ is translation by the antipodal map of $\R^d$. We further observe that the Wigner distribution \eqref{E:WignerK}  of 
$\tau \Pi_{\hbar, E_N(\hbar)} $ is \begin{equation}
 W_{\tau  \Pi_{\hbar, E_N(\hbar)}, \hbar}(x, \xi): = (2 \pi\hbar)^{-d} \int_{\R^d}  \Pi_{\hbar, E_N(\hbar)} \left( x+\frac{v}{2}, -x+\frac{v}{2} \right) e^{-\frac{i}{\hbar} v \xi} \frac{dv}{(2\pi h)^d},
\end{equation} 
and $$\mathrm{Tr}\; \tau \Pi_{\hbar, E_N(\hbar)}
=  (2 \pi\hbar)^{-d} \int_{T^*\R^d}   \Pi_{\hbar, E_N(\hbar)} \left( x+\frac{v}{2}, -x+\frac{v}{2} \right) e^{-\frac{i}{\hbar} v \xi} \frac{dv d \xi}{(2\pi h)^d},$$
which obviously is the same as \eqref{Q}.

It is well-known that asymptotically, the trace of a semi-classical Fourier
integral operator is to leading order given by the integral of its principal
symbol over the fixed point set of its canonical relation. The element
$\tau$ generates a $\Z_2$ subgroup acting on $T^*\R^d$ as $(x, \xi) \to (-x, -\xi) $, whose fixed point set is the origin $(0,0)$. In this case, \eqref{Q} gives an exact formula of this kind.

\section{Wigner distribution of the propagator: Proof of Propositions \ref{WIGNERLAGUERRE}, \ref{ucalprop-new}}\label{WIGPROP} 
\noindent Let $U_{\hbar}(t) = e^{\ihbar t \widehat{H}_h}$ denote the semi-classical propagator as in \eqref{PROPDEF}. Then its Wigner distribution is 
\begin{equation} \label{WIGU} 
\ucal_\hbar(t,x,\xi)= \int e^{- \ihbar 2 \pi  \xi\cdot v} U_{\hbar}\lr{t, x + \frac{v}{2}, x - \frac{v}{2}}\frac{dv}{(2\pi \hbar)^d}. \end{equation}
This may be put in a concrete form using the Mehler formula for $U_\hbar(t, x,y)$ in the position representation  \cite{F}:
\begin{equation}  \label{E:Mehler}
 U_\hbar(t, x,y) =  e^{-\ihbar t H_\hbar}(x,y)=  \begin{array}{l}  (2\pi i \hbar \sin t)^{-d/2}
 \exp\left( \frac{i}{\hbar}\left(
 \frac{\abs{x}^2 + \abs{y}^2}{2} \frac{\cos t}{\sin t} - \frac{x\cdot
 y}{\sin t} \right) \right),
 \end{array}
\end{equation}
where $t \in \R$ and $x,y \in \R^d$. The right hand side is singular at $t=0.$ It is well-defined as a distribution, however, with $t$ understood as $t-i0$. Indeed, since $H_h$ has a positive spectrum the propagator $U_h$ is holomorphic in the lower half-plane and $U_h(t, x, y)$ is the boundary value of a holomorphic function in $\{\Im t < 0\}$. In Section \ref{S:Wigner-prop-pf}, we prove Proposition \ref{ucalprop-new}. We combine the result of computation in Section \ref{S:Wigner-Laguerre-Pf} with a generating function identity for Laguerre functions to obtain the exact expression in Proposition \ref{WIGNERLAGUERRE} for the Wigner functions $W_{\hbar, E_N(\hbar)}$ for the spectral projection onto the $N^{th}$ eigenspace of $\widehat{H}_\hbar.$

\subsection{Proof of Proposition \ref{ucalprop-new}}\label{S:Wigner-prop-pf} 
In this section, we compute the Wigner transform $\ucal_{\hbar}(t,x,\xi)$ for the propagator $U_\hbar(t,x,\xi)$ of the isotropic harmonic oscillator $\widehat{H}_{\hbar}$ (see \eqref{Hh}). Our argument closely follows the startegy used by Thangavelu in \cite[Thm. 1.3.6]{T}. Since $\widehat{H}_\hbar$ is the sum of commuting $1$d oscillators, the scaling relation \eqref{E:Scaling Relation} yields
\begin{equation}
U_\hbar(t,x,y) =\prod_{j=1}^d \hbar^{-1/2}U(t,x_j/\hbar^{1/2}, y_j/\hbar^{1/2})\label{E:prop-prod}
\end{equation}
where the terms $U\equiv U_1$ in the product are propagators for the standard $1$d oscillator at $\hbar=1$. A simple change of variables in the definition of the Wigner transform 
\[W_{\hbar,K}(x,\xi)=\int_{\R^d}e^{-\frac{i}{\hbar}v\cdot \xi}~~
K\lr{x+\frac{1}{2}v,\, x-\frac{1}{2}v}\frac{dv}{\lr{2\pi \hbar}^d}\] 
of a kernel $K(x,y)$ shows that if we define
\[K_\beta(x,y):=\beta^{d/2}K(\beta^{1/2}x,\beta^{1/2}y),\qquad
x,y\in \R^d\]
then
\[W_{\hbar, K_\beta}(x,\xi)= W_{\hbar, K}(\beta^{1/2}x, \beta^{-1/2}\xi)=\hbar^{-d} W_{1,K}(\beta^{1/2}x, \hbar \beta^{-1/2} \xi).\]
Since the Wigner transform $K\mapsto W_{\hbar, K}$ preserves tensor products, the relation \eqref{E:prop-prod} together with the previous line yield
\[\ucal_{\hbar}(t,x,\xi)=\hbar^{-d}\prod_{j=1}^d\ucal_1\lr{t, \,x_j/\hbar^{1/2},\, \xi_j/\hbar^{1/2}}.\]
It therefore remains to compute the Wigner tranform $\ucal_1(t,x,\xi)$ for the standard $1$d oscillator $\widehat{H}_1$ with $\hbar=1$. To do this, recall that the spectrum of $\widehat{H}_1$ when $d=1$ is $N+1/2,\, N\geq 0$. Writing $h_N$ for the corresponding $L^2-$normalized Hermite functions, we have
\[U(t,x,y)=\sum_{N=0}^\infty e^{-it(N+1/2)}h_N(x)h_N(y).\]
Writing $r=e^{-it},$ the Mehler Formula reads
\begin{equation}
\sum_{N\geq 0}e^{-iNt} h_N(x)h_N(y)= \frac{1}{\sqrt{\pi}}(1-r^2)^{-1/2} \exp\left[-\frac{x^2+y^2}{2}\frac{1+r^2}{1-r^2} + \frac{2r}{1-r^2}xy\right].
\end{equation}
Hence, $U(t,x+v/2,x-v/2)$ is
\[r^{1/2} \sum_{N=0}^\infty r^Nh_N(x+v/2)h_N(x-v/2)=r^{1/2} \frac{1}{\sqrt{\pi}}(1-r^2)^{-1/2}\exp\left[-\frac{1+r}{1-r}\frac{v^2}{4}-\frac{1-r}{1+r}x^2\right].\]
By definition, the Wigner function $\ucal_1(t,x,\xi)$ is the Fourier transform of this expression in the $v$-variable with the normalization
\[\mathcal F (f)(\xi)=\int_{\R} e^{-iv\xi}f(v)\frac{dv}{2\pi}.\]
Using that 
\[\mathcal F(e^{-\alpha x^2})(\xi)= \frac{1}{2\sqrt{\pi \alpha}}e^{-\xi^2/4\alpha}\]
we find 
\[ \ucal_{1}(t,x,\xi)=\frac{r^{1/2}}{\pi(1+r)}\exp\left[-2H
  \frac{1-r}{1+r}\right],\qquad H=H(x,\xi)=\frac{\abs{x}^2+\abs{\xi}^2}{2}.\]
Substituting this into \eqref{E:prop-prod} completes the proof of Proposition \ref{ucalprop-new}.

\subsection{Proof of Proposition \ref{WIGNERLAGUERRE}}\label{S:Wigner-Laguerre-Pf}
Note that the Wigner function for the propagator and spectral projectors are related by
\[W_{\hbar, t}(x,\xi)=\sum_{N\geq 0}r^{N+d/2} W_{\hbar, E_N(\hbar)}(x,\xi),\qquad r=e^{-it}.\]
The usual generating function for associated Laguerre polynomials reads for every $\alpha>-1:$
\begin{equation*}
\sum_{N=0}^\infty L_N^{(\alpha)}(x) e^{-x/2 }r^N=(1-r)^{-\alpha-1}
e^{-\frac{1+r}{1-r}\frac{x}{2}}.
\end{equation*}
Combining this with Proposition \ref{ucalprop-new}, we have
\[W_{\hbar,t}(x,\xi)=\lr{\frac{r^{1/2}}{\pi \hbar (1+r)}}^de^{-\frac{2H}{\hbar}\frac{1-r}{1+r}}=\lr{\pi\hbar}^{-d}\sum_{N\geq 0}r^{N+d/2} (-1)^N L_N^{(d-1)}(4H/\hbar)e^{-2H},\]
 Hence, recalling that $r=e^{-it}$ and using the Fourier expansion 
\begin{equation}
W_{t,\hbar}(x,\xi)= \sum_{N\geq 0} e^{-it E_N(\hbar)/\hbar}W_{\hbar, E_N(\hbar)}(x,\xi)=\sum_{N\geq 0} \lr{e^{-it}}^{N+d/2}W_{\hbar, E_N(\hbar)}(x,\xi)\label{E:Fourier-expansion}
\end{equation}
we find
\[W_{\hbar, E_N(\hbar)}(x,\xi)=\lr{\pi \hbar}^{-d} (-1)^N L_N^{(d-1)}(4H/\hbar)e^{-2H/\hbar},\]
as claimed. 


\subsection{\label{WIGasDIST}Distribution properties of the Wigner distribution of
the propagator}
As noted in the Introduction, the Wigner distribution distribution of the propagator  is not a locally $L^1$ function, but is well-defined as a tempered distribution in the sense that the integral
        \begin{equation}\label{ucalhf}  \ucal_{\hbar, f}(x, \xi) : =\int_{\R} \hat{f}(t)  \ucal_{\hbar}(
         t, x, \xi)dt.  \end{equation}
is well-defined for $f \in \scal(\R)$. In fact, $\ucal_{\hbar}(t+\pi, x, \xi)$ resembles the position space propagator \eqref{E:Mehler} with $y = \xi$ except that its phase lacks
the term $- \frac{2}{\cos t/2} \langle x, \xi \rangle $. Note that for small $t$
the phase of \eqref{E:Mehler} is essentially $\frac{x^2 + y^2}{2 t} - \frac{1}{t}\langle x, y \rangle = |x - y|^2/2t$ and that as $t \to 0$ this kernel
weakly approaches $\delta(x - y)$. Similarly,  $\ucal_{\hbar} (t, x, \xi)
\to \delta_0(x, \xi)$ as $t \to  \pi$, the Dirac distribution at the point $(0,0) \in T^* \R^n$. Indeed, at $t=\pi$ it is the Wigner distribution of $\delta(x + y)$, which
is $\int_{\R^n} \delta(x - \frac{v}{2} + x + \frac{v}{2}) e^{- i \langle v, \xi \rangle} dv = \delta(2x) \delta(\xi). $ Since it is a locally $L^1$ function at times
$t \notin \pi \Z$ and is a measure when $t \in \pi \Z$ it is a measure for all $t$. We may consider it as a one-parameter family of measures of the variable $r^2= 2H(x,\xi)=||x||^2 +||\xi||^2 \in  \R_+$ with integration measure $r^{2d-1} dr$. We often fix  and think of $\ucal_{\hbar}(t, x, \xi)$ as a distribution in $t$:
\begin{equation}
 f \mapsto \ucal_{\hbar, f}(H) = \frac{1}{(2\pi\hbar)^d}  \int_{\R}(\cos(t/2- i0)^{-d}  \hat{f}(t) e^{-\frac{2Hi}{\hbar} \tan(t/2-i0)}\frac{dt}{2\pi},\quad H=\frac{\norm{x}^2+\norm{\xi}^2}{2}.\label{E:Dist-Wig-def}
\end{equation}
Note that this distribution is actually a smooth function away from $t\in \set{\pi, 3\pi, \ldots}.$ Let us now check that the contribution to the integral in \eqref{E:Dist-Wig-def} from a neighborhood of any of these points is $O(\hbar^\infty)$ provided $(x,\xi)$ is not too close to $(0,0).$ Since the Wigner function $\ucal_\hbar(t)$ is periodic, it suffices to check this near $t=\pi.$ To do this, for each $\delta>0$ define the smooth cut-off function $\chi_\delta:\R\gives [0,1]$ satisfying
\[\chi_\delta(t)=\begin{cases}
  1,&\quad \abs{t-\pi}<\delta\\
 0, &\quad \abs{t-\pi} > 2\delta
\end{cases}
.\]
\begin{lemma}\label{LOCLEM}
Uniformly over $(x,\xi)\in T^*\R^d$ with $H(x,\xi)>\hbar^2,$ the localized version
\begin{equation}
  \label{E:elevelloc-2}
  \int_{-\pi}^\pi e^{itE/\hbar} \chi_\delta(t)\widehat{f}(\pi/2+t)  \ucal_{\hbar}(t,x,\xi)\frac{dt}{2\pi}
\end{equation}
of \eqref{E:Dist-Wig-def} with $\delta = 4(H(x,\xi)/E)^{1/2}$ is $O((H(x,\xi)^{-1/2}\hbar)^\infty).$
\end{lemma}
\begin{proof}
The operator
\[D_{\hbar, \rho}(t):=-\frac{2i}{\hbar \rho} \cos^2(t/2)\dell_t\]
satisfies
\[D_{\hbar,\rho}(t) e^{-\frac{i}{\hbar}\tan(t/2)\rho}=e^{-\frac{i}{\hbar}\tan(t/2)\rho}.\]
It's adjoint is
\[D_{\hbar, \rho}^*(t):=\frac{2i}{\hbar \rho} \left[\cos^2(t/2)\dell_t~+~\frac{1}{2}\sin(t)\right],\]
and if $g(t)$ is any function that has a pole of order $k$ at $t=\pi$, then $D_{\hbar,\rho}g$ has a pole of order $k-1$ at $t=\pi$. Thus, integrating by parts using $\lr{D_{\hbar,\rho}(t)}^d,$ we find that the integral in \eqref{E:elevelloc-2} has the form
\begin{equation}  \label{E:elevelloc-3a}
 (2\pi \hbar^{-d}) \int_{-\pi/2}^{\pi/2} A_{\hbar,\rho}(t) e^{\frac{i}{\hbar}\Psi_{\rho, E}(t)}\frac{dt}{2\pi}, \qquad \Psi_{\rho,E}(t)=tE-\tan(t/2)\rho
\end{equation}
where $A_{\hbar,\rho}(t)$ is a smooth function (including at $t=\pi$) and is uniformly bounded for $\abs{t}<\pi/4$ and for all $\hbar.$ Indeed, the smoothness of $A_{\hbar,\rho}$ follows from the observation above that each application of $D_{\hbar, \rho}^*(t)$ increments the order of vanishing of the amplitude at $t=\pi$ by $1$. The uniform boundedness follows from the obervation that each application of $D_{\hbar, \rho}(t)$ contains both an $\hbar$ (from the prefactor $D_{\hbar, \rho}(t)$) and an $\hbar^{-1}$ when the time derivative is applied the rapidly oscillating term $e^{itE/\hbar}.$

To complete the argument, note that the time derivative $-H(x,\xi)/\cos^2(t/2) + E$ of the combined phase $\Psi$ is uniformly bounded away from $0$ in the interval $[\pi-\delta, \pi+\delta]$ for
\[\delta = \delta(x,\xi)=4(H(x,\xi)/E)^{1/2}.\] 
Hence on the support of $\chi_\delta$ we may consider the operator
\[P(t):=\frac{\hbar}{i \dell_t \Psi(t)}\dell_t,\qquad P^*(t) = - \frac{\hbar}{i}\lr{\frac{1}{E-H(x,\xi)/\cos^2(t/2)}\dell_t + \frac{H(x,\xi)/\sin(t/2)\cos^3(t/2)}{\lr{E-H(x,\xi)/\cos^2(t/2)}^3}}\]
and its adjoint. Repeated integration by parts using $P(t)$ in \eqref{E:elevelloc-3a} shows that the integral is therefore of size $O((H(x,\xi)^{-1/2}\hbar)^\infty)$.
\end{proof}

\section{Interface scaling of individual eigenspace Wigner distributions: Proof of Theorem \ref{SCALINGCOR-old}}
In this section, we study the scaling asymptotics  of Wigner distributions $W_{\hbar_N, E_N(\hbar)}(x, \xi)$ of individual eigenspace projections along the interface $\Sigma_E$.  We essentially give two proof of Theorem \ref{SCALINGCOR-old}, one using special facts about Laguerre functions, and the second using semi-classical analysis.

\subsection{Semi-classical proof of Theorem \ref{SCALINGCOR-old}}\label{SCSCORSECT} 
 We have by Proposition \ref{ucalprop-new} that
\[W_{\hbar,E_N(\hbar)}(x,\xi):=\int_{-\pi}^\pi A_\hbar(t)\exp\left[\frac{i}{\hbar}\Psi(t)\right]\frac{dt}{2\pi},\]
where
\[A(t):=\lr{2\pi \hbar\cos(t/2)}^{-d}\]
and
\begin{equation} \label{PSIDEF}\Psi(t)=\Psi_{E,H}(t)=-2H\tan\lr{\frac{t}{2}}+tE.\end{equation}
The critical point equation $\dell_t \Psi=0$ is,
\begin{equation} \label{CPE2} \cos^2\lr{t/2}=H/E.\end{equation}
Hence, when $H\approx E,$ we will have a degenerate critical point at $t=0,$ causing Airy-like behavior. To see this more precisely, we write as in the statement of Theorem \ref{SCALINGCOR-old} 
\[H=H(x,\xi)=E - u\lr{\frac{\hbar}{2E}}^{2/3}.\]
Since there are no critical points of $\Psi(t)$ outside $\abs{t}\gg \hbar^{1/3}|u|^{1/2}$, we may use Lemma \ref{LOCLEM} and integration by parts to write
\begin{equation}
W_{\hbar, E_N(\hbar)}(x,\xi)= \int_{-\pi}^\pi \chi(\hbar^{-1/3+\epsilon} \abs{u}^{-1/2}t) A_\hbar(t)\exp\left[\frac{i}{\hbar}\Psi(t)\right]\frac{dt}{2\pi}+O(\hbar^\infty),\label{E:wig-unscaled}
\end{equation}
where $\chi:\R\gives [0,1]$ is a smooth cut-off function that is identically $1$ on $[-1,1]$ and vanishes outside $[-2,2].$ Taylor expansion gives
\[\Psi(t\hbar^{1/3})= t\hbar^{1/3}\lr{E-H} - \frac{\hbar H}{4}\frac{t^3}{3} + O(\hbar^{5/3}t^5),\quad A(t\hbar^{1/3})=(2\pi\hbar)^{-d}\lr{1+O(t^2\hbar^{2/3})}.\]
Hence, we obtain after changing variables $T=t (\hbar E/4)^{1/3}$ in \eqref{E:wig-unscaled}
\[W_{\hbar, E_N(\hbar)}(x,\xi)= \frac{\hbar^{1/3}}{(2\pi \hbar)^{d}} C_E \int_{\R}\chi(C_E\hbar^{\epsilon} T\abs{u}^{-1/2}) \cos\lr{\frac{Tu}{E} + \frac{T^3}{3}} \frac{dT}{2\pi}~+~O(\hbar^{1-d-6\epsilon}\abs{u}^{3}),\]
where $C_E=(4/E)^{1/3}.$ We split the integral over $T$ into two pieces:
\[\int_{\R}\cos\lr{\frac{Tu}{E}+\frac{T^3}{3}}\frac{dt}{2\pi}+\int_{\abs{t}\geq \hbar^{-\epsilon}\abs{u}^{1/2}}(\chi(\hbar^{\epsilon}\abs{u}^{1/2}t)-1)\cos\lr{\frac{Tu}{E}+\frac{T^3}{3}}\frac{dt}{2\pi}.\]
The first term is precisely $\Ai(u)$. The second term is $O(\hbar^\infty)$ because one can integrate by parts repeatedly using
\[-\lr{T^2+\frac{u}{E}}^{-1}\dell_T \cos\lr{\frac{Tu}{E}+\frac{T^3}{3}}=\sin\lr{\frac{Tu}{E}+\frac{T^3}{3}}\]
and the similarly identity for $\sin$ to obtain an arbitrary power of $T$ in the denominator of the integrand. This completes the proof of Theorem \ref{SCALINGCOR-old}. \hfill $\square$

\subsection{Laguerre proof of Theorem \ref{SCALINGCOR-old}
 and Proposition \ref{EXTDECAY}}\label{Laguerre-Scaling-Proof} In this section, we prove \eqref{E:W-scaling} by appealing to the following uniform Airy asymptotics of Laguerre functions.

\begin{proposition} \label{FW} [Frenzen and Wong, \cite{FW} Section 5 (5.12)]
Fix $b\leq s <\infty$, where $0<b<1$ is fixed. For each $n\geq , \alpha\geq 0,$ write
\[ \nu = 4n+2\alpha + 2.\] 
We have
\begin{equation} \label{LAIRY2} \begin{array}{lll}  (-1)^n 2^{\alpha} e^{- \nu s/2} L_n^{\alpha} (\nu s) 
 &= & \Ai (\nu^{2/3} B^2(s)) \alpha_0(s)+ \epsilon_1(s , \nu), \end{array} \end{equation}
where 
\begin{equation}\label{BFORMULA} B(s) = \left\{\begin{array}{ll}  i (3 \beta(s)/2)^{1/3}, & s \in (0,1];
\\&\\
(3 \gamma(s)/2)^{1/3}, & s \in [1, \infty). \end{array}\right. \end{equation}
and
\begin{equation} \label{beta} 
\begin{array}{l} \beta(s) = \half [\cos^{-1} \sqrt{s} - \sqrt{s - s^2}], \\ \\
\gamma(s) = \half [\sqrt{s^2 - s} - \cosh^{-1}\sqrt{s} ]. \end{array} \end{equation}
We have
\begin{equation}
 \label{E:FW_alpha0} \alpha_0(s) =
 \begin{cases}
s^{(1-\alpha)/2} \frac{\sqrt{2|B(s)|}}{(1-s)^{1/4}s^{3/4}},&\quad 0<s<1\\
s^{(1-\alpha)/2} \frac{\sqrt{2B(s)}}{(s-1)^{1/4}s^{3/4}},&\quad s\geq 1
 \end{cases}.
\end{equation}
The error term $\epsilon_1(s,\nu)$ satisfies
\begin{equation} \label{eq:remainder} 
\epsilon_1(s,\nu) \leq 
  C_p \nu^{-5/3}\abs{\twiddle{\beta}_1(s)}\abs{\twiddle{\Ai}'(\nu^{2/3}B^2(s))},
\end{equation}
where 
\[ \twiddle{\Ai}(z) = \begin{cases} \Ai(z) & \text {if $z \geq 0$} \\ [\Ai(z)^2 + \mathrm{Bi}(z)^2]^{1/2} & \text{if $z<0$. } \end{cases} \]
\[
 \twiddle{\Ai}'(z) = \begin{cases} \Ai'(z) & \text {if $z \geq 0$} \\ [\Ai'(z)^2 + \mathrm{Bi}'(z)^2]^{1/2} & \text{if $z<0$. } \end{cases} \]
and finally $\twiddle{\beta}_1(s)$ is a continuous function of $s$ that is uniformly bounded on any compact subset of $s>0.$
\end{proposition}
Let us first recall the notation. We fix $E>0$ and define for each $N,d\geq 1$
\[\hbar=\hbar(E,N)=\frac{E}{N+d/2}.\]
With this notation, $E$ belongs to the spectrum of the harmonic oscillator with $\hbar(E,N)$ and any $N.$ Thus, $W_{\hbar, E_N(\hbar_N)},$ which we abbreviate $W_{\hbar, E}$ is the Wigner function for the projection onto the eigenspace of energy $E$, and we seek to study its scaling asymptotics around the energy surface $\Sigma_E=\set{(x,\xi)\in T^*\R^d\,|\,\abs{x}^2+\abs{\xi}^2 = 2E}.$ To do this, recall from Proposition \ref{WIGNERLAGUERRE} that 
\[W_{\hbar,N}(x,\xi)=\frac{1}{\lr{\pi \hbar}^d} (-1)^N
L_N^{(d-1)}(4H/\hbar)e^{-2H/\hbar},\quad
H=H(x,\xi)=\frac{\abs{x}^2+\abs{\xi}^2}{2}.\]
Note that the parameter $\nu$ from Proposition \ref{FW} corresponding to this Laguerre function is
\[ \nu = 4N+2(d-1) + 2=\frac{4E}{\hbar}= \frac{2}{\hbar_E},\]
where we've introduced
\[\hbar_E:=\frac{\hbar}{2E}.\]
Consider $(x,\xi)\in T^*\R^d$ with
\[H(x,\xi)=\frac{\abs{x}^2+\abs{\xi}^2}{2}=E - u \hbar_E^{2/3}.\]
We find 
\[W_{\hbar,N}(x,p)=\frac{2}{\lr{2\pi \hbar}^d} (-1)^N 2^{d-1}
L_N^{(d-1)}(\nu s)e^{\nu s/2},\qquad \nu = \frac{2}{\hbar_E},\,\, s =
\frac{H}{E}.\] \
Proposition \ref{FW} therefore yields
\begin{align}
\label{WAIRY}  W_{\hbar, N}(x,\xi)&=\frac{2}{(2\pi \hbar)^d} \left[\nu^{-1/3}\Ai(\nu^{2/3}B^2(s))\alpha_0(s) +\epsilon_1(\nu,s)\right].
\end{align}
To simply this expansion, we need the following estimates, obtained by direct Taylor expansion of the definitions \eqref{beta} and \eqref{E:FW_alpha0},
\begin{equation}
B^2(1+t)=2^{-2/3}t(1+ O(t)),\qquad \alpha_0(1+t)=2^{1/3}(1+O(t)).\label{E:alpha-B-exp}
\end{equation}
These estimates yield
\begin{equation}
\nu^{2/3}B^2\lr{s}= \lr{\frac{2}{\hbar_E}}^{2/3}B^2(1+(u/E)\hbar_E^{2/3})= \frac{u}{E} + O\lr{\abs{u}^2\hbar_E^{2/3}}. \label{E:B-exp}
\end{equation}
Hence, using that $\beta_1$ is continuous and that as $u\gives -\infty$
\begin{equation}
\sup_{u<0} \frac{\abs{\twiddle{\Ai}'(u/E)}}{\lr{1+\abs{u}}^{1/4}}<\infty,\label{E:Airy-growth}
\end{equation}
we find that for $u<0$
\begin{align}
\abs{\epsilon_1(\nu,s)} =O \lr{\hbar_E^{5/3}\twiddle{\Ai}'\lr{u/E + O\lr{\abs{u}^2\hbar_E^{2/3}}}}= O\lr{\hbar_E^{5/3}\lr{1+\abs{u}}^{1/4}},\label{E:error-uneg-1}
\end{align}
with the implied constants being uniform when $u^2\hbar_E^{2/3}$ ranges over a compact set. Similarly, using that 
\begin{equation}
\sup_{u>0} \frac{\abs{\twiddle{\Ai}'(u/E)}/\abs{\Ai(u/E)}}{\lr{1+\abs{u}}^{1/2}}<\infty,\label{E:Airy-growth-2}
\end{equation}
we find that with $s=1+(u/E)\hbar_E^{2/3}$ and $u>0$
\begin{align}
\abs{\epsilon_1(\nu,s)}= \Ai\lr{\nu^{2/3}B^2(s)}\cdot O \lr{\hbar_E^{5/3}\lr{1+\abs{u}}^{1/2}},\label{E:error-upos-1}
\end{align}
again with the implied constant uniform when $u^2\hbar_E^{2/3}$ ranges over a compact set. The estimates \eqref{E:alpha-B-exp} also yield
\begin{equation}
\nu^{-1/3}\Ai(\nu^{2/3}B^2(s))\alpha_0(s)=\hbar_E^{1/3}\Ai\lr{\hbar_E^{-2/3}B^2\lr{s}}\lr{1+O\lr{(s-1)^{2/3}}}.\label{E:leading-term-1}
\end{equation}
Using \eqref{E:B-exp} and the estimates \eqref{E:Airy-growth} and \eqref{E:Airy-growth-2}, we therefore find
\[\nu^{-1/3}\Ai(\nu^{2/3}B^2(s))\alpha_0(s)=\hbar_E^{1/3}\left[\Ai(u/E)+O((1+\abs{u})^{1/4}\abs{u}\hbar_E^{2/3})\right]\]
when $u<0$ and 
\[\nu^{-1/3}\Ai(\nu^{2/3}B^2(s))\alpha_0(s)=\hbar_E^{1/3}\Ai(u/E)\left[1+O((1+\abs{u})^{3/2}u\hbar_E^{2/3})\right]\]
when $u>0.$ Putting these estimates together with the error estimtes \eqref{E:error-uneg-1} and \eqref{E:error-upos-1} above completes the proof of Theorem \ref{SCALINGCOR-old}. In fact, \eqref{WAIRY} together with \eqref{eq:remainder} yield the following more precise result. 

Moreover, for $H(x,\xi)\geq E$, we have the more refined estimate
\begin{theo}\label{AIRYCOR}
\begin{equation}
\label{WIGNERASYMPintro-forbidden}
W_{\hbar, E_N(\hbar)}(x, \xi)  \simeq \frac{2\nu^{1/3} }{(2 \pi \hbar)^d}\Ai(\nu^{2/3}( B^2(H/E))) \alpha_0(H/E)\lr{1+O(\hbar)},\quad \nu = \hbar/4E_N(\hbar).
\end{equation}
\end{theo}

\subsection{\label{EXTDECAYSECT} Exterior exponential decay: Proof of Proposition \ref{EXTDECAY}}
The exponential decay in the exterior of $\Sigma_E$ stated in Propostion \ref{EXTDECAY} follows from Theorem \ref{AIRYCOR} and from the case $s \geq 1$ of Proposition \ref{FW}. 
\begin{proof}
We continue to write $E_N(\hbar)  =E$, and
\[H_E=\frac{H(x,\xi)}{E} > 1,\qquad \nu= 4E/\hbar.\]
Using the Airy asymptotics on the positive real axis in \eqref{AIRYASYM}, we find as $\hbar \to 0$, up to a multiplicative $1+O(\hbar)$ factor, that
\begin{align}
  \label{WIGNERASYMPintro2} 
W_{\hbar, E_N(\hbar)}(x, \xi) &  = \frac{2}{(2 \pi \hbar)^d} \nu^{-1/3} \Ai(\nu^{2/3}B^2(H_E)) \alpha_0(H_E)  \\
\notag & \simeq  \frac{\nu^{-1/3}\lr{\nu^{2/3}B^2(H_E)}^{-1/4} \alpha_0(H_E) }{4\pi^{3/2}(2\pi \hbar)^{d}}  e^{- \frac{2}{3} \lr{\nu B^3(H_E)}}.
\end{align}
By \eqref{BFORMULA} and by \eqref{beta}, we have
$$B(H_E)^3 = \frac{3}{4}\left[\sqrt{H_E^2 - H_E} - \cosh^{-1}\sqrt{H_E} \right].$$  
Also, by \eqref{E:FW_alpha0}, $\alpha_0(H_E)$ is decreasing as $H_E$ grows, and $B(H_E)$ is bounded by a constant times $H_E$ when $H_E>1$. It follows that there exists a constant $C_1>0$ so that
$$|W_{\hbar, E_N(\hbar)}(x, \xi)| \leq C_1 \hbar^{-d + \frac{1}{2}} 
e^{- \frac{2 E}{\hbar}[\sqrt{H_E^2 -H_E} - \cosh^{-1}\sqrt{H_E} ]},\qquad H_E > 1. $$
Finally, as $H_E \to \infty$, 
 \[\sqrt{H_E^2 - H_E} - \cosh^{-1}\sqrt{H_E}=H_E\lr{1+O(H_E^{-1}\log(H_E))}.\] 
Hence, for $H_E $ sufficiently large,
$$|W_{\hbar, E_N(\hbar)}(x, \xi)| \leq C_1 \hbar^{-d + \frac{1}{2}} 
e^{- \frac{2H}{\hbar}}. $$
\end{proof}

\subsection{\label{BESSELSECT} Bessel and Cosine Asymptotics: Laguerre proof of Proposition \ref{BESSELPROP}}
As before, we fix $E>0$ and set 
\begin{equation}
E=E_N(\hbar)=\hbar(N+d/2).\label{E:fixed-energy}
\end{equation}
We seek to study the behavior of $W_{\hbar, E_N(\hbar)}(x,\xi)$ inside the energy surface but not at the origin, i.e. under the constraint
\[0<H_E= \frac{\abs{x}^2+\abs{\xi}^2}{2E}<1.\]
To do this, we begin with the exact formula 
\begin{equation}
W_{\hbar, E_N(\hbar)}(x,\xi)=\frac{(-1)^N}{(\pi \hbar)^d}e^{-\nu H_E/2} L_N^{(d-1)}(\nu H_E) ,\qquad \nu = \frac{4E}{\hbar}.\label{E:exact-wig}
\end{equation}
Note that the degree $N$ of the Laguerre function grows like $\hbar^{-1}$ under the constraint \eqref{E:fixed-energy}. We will use the following expansion.

\begin{proposition} \label{FW2} [Frenzen and Wong, \cite{FW} Section 5 (5.12)]
Fix an integer $p \geq 1$ let $0\leq t <1$. For each $n\geq , \alpha\geq 0,$ write
\[ \nu = 4n+2\alpha + 2.\] 
We have
\begin{equation} \label{LBESS2}  2^{\alpha} e^{- \nu t/2} L_n^{\alpha} (\nu t) 
 =  \frac{J_{\alpha}(\nu A(t))}{A(t)^{\alpha}}\alpha_0(t) + \epsilon_1(t , \nu),  \end{equation}
where 
\begin{equation}\label{AFORMULA} A(t) = \frac{1}{2}\lr{\sin^{-1}(\sqrt{t}) + \sqrt{t-t^2}} \end{equation}
and
\begin{equation}
 \label{E:FW_alpha0_bes} \alpha_0(t) = \lr{\frac{A(t)}{\sqrt{t}}}^{\alpha+1} \lr{\frac{t}{1-t}}^{1/4}(A(t))^{-1/2}.
\end{equation}
The error term $\epsilon_1(s,\nu)$ satisfies for each $0\leq a<1$
\begin{equation} \label{eq:remainder_bess} 
\sup_{0\leq t <a}\abs{\epsilon_1(t,\nu) }\leq C_{a} \nu^{-1}\frac{\abs{J_{\alpha+1}(\nu A(t))}}{A(t)^{\alpha+1}}
\end{equation}
\end{proposition}

 Let $\alpha=d-1$ and $\nu = 4n + 2 \alpha +2=4E/\hbar$ and let $0 < a < 1$.  The Bessel expansion in the previous Proposition combined with \eqref{E:exact-wig} yields:
 \[W_{\hbar, E_N(\hbar)}(x, \xi) = (-1)^N\frac{2}{(2\pi\hbar)^d}\left[\frac{J_{d-1}(\nu A(H_E))}{A(H_E)^{d-1}}\alpha_0(H_E)+O\lr{ \nu^{-1}\abs{\frac{J_{d}(\nu A(H_E))}{A(H_E)^{d}}}}\right],\]
where the implied error term is uniform uniform over $H_E\leq a.$ Using the definition \eqref{E:FW_alpha0_bes} of $\alpha_0,$ we find
\[\frac{\alpha_0(H_E)}{A(H_E)^{d-1}}=\frac{A(H_E)^{1/2}}{H_E^{d/2}} \lr{\frac{H_E}{1-H_E}}^{1/4}.\]
Combining this with the standard Bessel asymptotics 
\[J_k(t)\sim \sqrt{2/(\pi t)}\cos(t-\tfrac{\pi}{2}(k+1/2))\lr{1+O(t^{-1})},\quad \text{as }t\gives\infty,\]
we obtain
\[W_{\hbar, E_N(\hbar)}(x, \xi) = (-1)^N\frac{2H_E^{d/2 - 1/4}}{(1-H_E)^{1/4}(2\pi\hbar)^d}\lr{\frac{\hbar}{2\pi E }}^{1/2}\cos\lr{\nu A(H_E)-\frac{\pi}{2}(d-1/2)}+O\lr{\hbar^{-d+3/2}},\]
where the implied constant is uniform when $H_E$ ranges over compact subsets of $(0,1).$ To complete the proof, note that since $E=\hbar(N+d/2)$ and $H_E=H/E$, we have
\begin{align*}
\nu A(H_E)-\frac{\pi}{2}(d-1/2)&=\frac{2E}{\hbar}\lr{(H_E - H_E^2)^{1/2} + \sin^{-1}(H_E^{1/2})}-\frac{\pi}{2}(d-1/2)\\
&=\frac{2E}{\hbar}\lr{H_E(H_E^{-1} - 1)^{1/2} + \frac{\pi}{2}- \cos^{-1}\lr{H_E^{1/2}}}-\frac{\pi}{2}(d-1/2)\\
&=N\pi +\frac{\pi}{4} +\frac{2H}{\hbar}(H_E - 1)^{1/2} -\frac{2E}{\hbar} \cos^{-1}\lr{H_E^{1/2}}.
\end{align*}
Thus, we obtain 
\[W_{\hbar, E_N(\hbar)}(x, \xi) = \frac{2 H_E^{-d/2}}{(H_E^{-1}-1)^{1/4}(2\pi\hbar)^d}\lr{\frac{\hbar}{2\pi E }}^{1/2}\cos\lr{\xi_{\hbar, E, H}}~+~O\lr{\hbar^{-d+3/2}},\]
where 
\[\xi_{\hbar, E, H}= \frac{\pi}{4} +\frac{2H}{\hbar}(H_{E}^{-1} - 1)^{1/2} -\frac{2E}{\hbar} \cos^{-1}\lr{H_E^{1/2}}.\]
This completes the proof of Proposition \ref{BESSELPROP}. \hfill $\square$

\subsection{Semiclassical Proof of \eqref{E:interior-cosine}}\label{S:sc-cosine-pf}

In this section, we give a separate semiclassical proof of \eqref{E:interior-cosine}. To do this, we express  $W_{U_\hbar(t)}(\rho)$ from \ref{WIGasDIST} in the WKB form,
\[W_{U_\hbar(t)}(\rho)=\ucal_{\hbar}(t, x, \xi) = A_\hbar(t)\exp\lr{\frac{i}{\hbar}\Psi(t,\rho)},\]
where 
\[A_\hbar(t)=\lr{2\pi \hbar \cos(t/2)}^{-d},\qquad \Psi(t,\rho)=-2 H_E \tan(t/2)+tE\]
and
\[H_E=\frac{H(x,\xi)}{E}=\frac{\norm{x}^2+\norm{\xi}^2}{2E}\]
as before. The Fourier expansion \eqref{E:Fourier-expansion} and Proposition \ref{ucalprop-new} therefore allow us to write
  \begin{equation}
W_{\hbar, E_N(\hbar)}(x,\xi)= \int_{-\pi}^\pi A_\hbar(t) \exp\left[\frac{i}{\hbar}\Psi_{H, E}(t)\right]\frac{dt}{2\pi}.\label{E:Fourier-Wigner}
\end{equation}
Fix $\epsilon\in (0,1/2)$. We will show that the relation \eqref{E:interior-cosine} holds uniformly for $(x,\xi)$ satisfying
\[\epsilon  < H_E< (1-\epsilon)E.\]
To do this, note that there exists $\delta=\delta(\epsilon)>0$ so that 
\[t\in [-\pi, \,-\pi+\delta]\cup[\pi-\delta, \,\pi]\quad \Rightarrow \quad \cos^2(t/2) < 1-\epsilon.\]
For this delta, fix a smooth cut-off function $\chi:[-\pi,\pi]\gives [0,1]$ so that
\[\chi(t)=
\begin{cases}
  1,&\quad t\in [-\pi+\delta/2,\, \pi-\delta/2]\\
  0,&\quad t\not\in [-\pi+\delta/4,\, \pi-\delta/4]
\end{cases}.
\]
Using Lemma \ref{LOCLEM} and \eqref{E:Fourier-Wigner}, we have
  \begin{equation}
W_{\hbar, E_N(\hbar)}(x,\xi)= \int_{-\pi}^\pi \chi(t) A_\hbar(t) \exp\left[\frac{i}{\hbar}\Psi_{H, E}(t)\right]\frac{dt}{2\pi}+O(\hbar^\infty).\label{E:Fourier-Wigner-Loc}
\end{equation}
The critical point equation 
\[\dell_t\Psi_{\rho, E}=0\qquad \Longleftrightarrow\qquad \cos^2(t/2)=H_E\] 
has four solutions
\[t_{\pm}=\pm 2 \cos^{-1}\lr{H_E^{1/2}},\qquad T_{\pm}=\pm 2 \cos^{-1}\lr{-H_E^{1/2}}.\]
However, only $t_{\pm}/2$ lie in the interval $[-\pi/2,\pi/2]$ over which we are integrating. By construction, $t_{\pm}$ lie in the interval where $\chi$ is identically equal to $1.$ A simple computation shows
\[\Psi_{\rho, E}(t_{\pm})=2\lr{\mp H\sqrt{H_E^{-1}-1} +\frac{t_{\pm}}{2}E},\qquad \dell_{tt}\Psi_{\rho,E}(t_{\pm})=\mp E\sqrt{\frac{E}{H}-1},\]
verifying that these critical points are non-degenerate when $H<E.$ We may therefore apply the method of stationary phase (Lemma \ref{L:SP LO}) to the integral \eqref{E:Fourier-Wigner-Loc}. Before writing the formula, note that 
\[e^{\frac{i\pi}{4}\sgn(\dell_{tt}\Psi(t_{+})) +\frac{i}{\hbar}\Psi(t_{+})}A_{\hbar}(t_+)=\overline{e^{\frac{i\pi}{4}\sgn(\dell_{tt}\Psi(t_{-})) +\frac{i}{\hbar}\Psi(t_{-})}A_{\hbar}(t_-)}\]
are complex conjugates. Thus, we may take real parts in the stationary phase expansion to write $W_{\hbar, E_N(\hbar)}(x,\xi)$ as an overall prefactor
\[(2\pi\hbar)^{-d+1/2} P_{E,H}  ,\quad P_{E,H}:=\lr{\pi E^{1/2}\lr{\frac{E}{H}-1}^{1/4}\lr{\frac{H}{E}}^{d/2}}^{-1}\]
times the contribution from $t_{\pm}:$
\[\cos\left[-\frac{\pi}{4} +\frac{2E}{\hbar}\lr{-\frac{H}{E}\lr{\frac{E}{H}-1}^{1/2}+\frac{t_{+}}{2}}\right]\]
plus an $O(\hbar)$ error. Stationary phase asymptotics of Lemma \ref{L:SP LO} complete the proof.  \hfill $\square$

\subsection{Proof of Proposition \ref{SMBALLS}}

Fix $\epsilon>0$ and a symbol $a(x,\xi) \in C^\infty (T^*\R^d).$ We wish to show that for $\epsilon>0$ sufficiently small 
   \begin{equation}
\int_{T^*\R^d} a(x, \xi) 
W_{\hbar, E_N(\hbar)}(x,\xi) \psi_{\epsilon,\hbar}(x,\xi)d x d \xi=O(\hbar^{\frac{1-d}{2}-2d\epsilon}\norm{a}_{L^\infty(B_0(\hbar^{1/2-\epsilon}))}),\label{E:localized2}
\end{equation}
where $\psi_{\epsilon,\hbar}$ is a smooth radial cut-off that is identically $1$ on the ball of radius $\hbar^{1/2-\epsilon}$ and is identically $0$ outside the ball of radius $2\hbar^{1/2-\epsilon}.$ 

\begin{proof} The Laguerre formula from Proposition \ref{WIGNERLAGUERRE} shows that the above integral, in polar coordinates $(x,\xi)\mapsto (r,\omega)$,  is
\[\int_{\R^{2d}}\left[\frac{(-1)^N}{\pi^d} L_N^{(d-1)}(2r^2)e^{-r^2}\right]\left[\psi_{\epsilon,\hbar}(\hbar^{-1/2}r) \int_{S^{d-1}} a(\hbar^{-1/2}r\omega) d\omega\right] r^{2d-1}dr,\qquad r^2=\norm{x}^2+\norm{\xi}^2.\]
Applying Cauchy-Schwarz, shows that the absolute value of this expression is bounded above by a constant times
\begin{equation}
\lr{\int_0^\infty \abs{ L_N^{(d-1)}(2r^2)e^{-r^2}}^2 r^{2d-1}dr}^{1/2}\hbar^{-d}\norm{\chi_\hbar a}_{L^2(T^*\R^d)}.\label{E:loc-bound}
\end{equation}
The integral on the left can be computed exactly using the usual orthogonality relations for Laguerre functions. Indeed, up to a constant we depending only on $d$, we find
\[\int_0^\infty \abs{ L_N^{(d-1)}(2r^2)e^{-r^2}}^2 r^{2d-1}dr\simeq \int_0^\infty \abs{ L_N^{(d-1)}(\eta)}^2e^{-\eta} \eta^{d-1}d\eta=\frac{(N+d-1)!}{N!}\simeq N^{d-1}\simeq \hbar^{-d+1}.\]
On the other hand, 
\[\hbar^{-d}\norm{\psi_{\epsilon,\hbar} a}_{L^2(T^*\R^d)}=O\lr{ \norm{a}_{L^\infty(B_0(\hbar^{1/2-\epsilon})} \hbar^{-2d\epsilon}},\]
where $B_0(\hbar^{1/2-\epsilon})$ is the ball of radius $2\hbar^{1/2-\epsilon}$ centered at $(0,0)\in T^*\R^d.$ Thus, using \eqref{E:loc-bound}, we find that
\[\int_{T^*\R^d} a(x, \xi) 
W_{\hbar, E_N(\hbar)}(x,\xi) \psi_{\epsilon,\hbar}(x,\xi)d x d \xi= O(\hbar^{\frac{1-d}{2}-2d\epsilon}\norm{a}_{L^\infty(B_0(\hbar^{1/2-\epsilon}))}),\]
confirming \eqref{E:localized}. 
\end{proof}

\subsection{\label{PROOFPROPOPWa}  Proof of Proposition \ref{OPWa}}
We now use the semi-classical asymptotics of Section \ref{S:sc-cosine-pf} to prove Proposition \ref{OPWa}.

 \begin{proof} In view of Proposition \ref{SMBALLS}, it suffices to study
\begin{equation}
\int_{T^*\R^d} (1-\psi_{\hbar, \epsilon}(x,\xi))a(x, \xi) W_{\hbar, E_N(\hbar)}(x,\xi) d x d \xi,\label{E:weak-star-main}
\end{equation}
where we recall that $\psi_{\epsilon,\hbar}$ is a smooth radial cut-off that is identically $1$ on the ball of radius $\hbar^{1/2-\epsilon}$ and is identically $0$ outside the ball of radius $2\hbar^{1/2-\epsilon}.$ Note that by the exponential decay of $W_{\hbar, E}(x,\xi)$ on $\set{H(x,\xi)>E}$ (see e.g. Proposition \ref{EXTDECAY}), we may and shall assume that the integrand is compactly supported in the $(x,\xi)$ up to $O(\hbar^\infty)$ errors. To study \eqref{E:weak-star-main}, we use \eqref{E:Fourier-Wigner} to rewrite it modulo a term of size $O(\hbar^\infty)$ as follows:
  \begin{equation} 
\int_{T^*\R^d}   \int_{-\pi}^\pi\chi_\delta(t) (1-\psi_{\epsilon,\hbar}(x,\xi)) A_\hbar(t) \exp\left[\frac{i}{\hbar}\Psi(t)\right] a(x, \xi)  \frac{dx d \xi dt}{2\pi},\label{E:Fourier-Wignera}
\end{equation} 
where $\Psi$ is defined by \eqref{PSIDEF} and $\chi_\delta$ is defined as in Lemma \ref{LOCLEM}. The critical point equation, $d_{x, \xi, t} \Psi_{\rho, E} = 0$ becomes
$$ \tan(t/2) d_{x, \xi} H_E(x, \xi) =0, \;\; \cos^2\lr{ \frac{t}{2}} = H_E.  $$
The solutions are $\{t = 0, H_E = 1\}$ or $\{(x, \xi) =0, t= \pi\}$, but by construction the second solution is not in the support of $1-\psi_{\epsilon,\hbar}$. The solution $H_E(x, \xi)  = 1, t = 0$ corresponds to $\{t = 0\} = \{0\} \times \Sigma_E$ and thus gives a non-degenerate critical manifold of dimension $2d -1$ in a parameter space of dimension $2d$. The phase vanishes on the critical point set and a simple computation shows that the normal Hessian of the phase in the radial variable $r$ on $T^*R^d$ and $t$ is
$\lr{\begin{array}{cc}
  0& - 1\\ -1 & 0 
\end{array}}.$
The derivatives of the amplitude are dominated by the $\hbar^{-1/2+\epsilon}$ sized derivatives of $\psi_{\epsilon, \hbar}$. The leading order term from stationary phase is 
$$\hbar^{-d} C_d\hbar \int_{\Sigma} a d\mu_L \simeq C_d \hbar^{-d +1}  \int_{\Sigma} a d\mu_L  $$
for a dimensional constant $C_d$, and the next term is on the order of $\hbar^{-d+1+2\epsilon}=o(\hbar^{-d+1}),$ as desired. \end{proof}

\subsection{Indirect proof  of Proposition \ref{OPWa}  }
 We can give a `softer' proof using the eigenvalue equation satisfied
 by $W_{\hbar, E_N(\hbar)}$.

\begin{proof}      From \eqref{TRACEP} and \eqref{INTEGRALS}(i)
we have the trace identity $$\frac{1}{\dim V_{\hbar, E_N(\hbar)}} \mathrm{Tr} Op_h^w(a) \Pi_{\hbar, E_N(\hbar)} = \frac{1}{\dim V_{\hbar, E_N(\hbar)}} \int_{T^* \R^d} a(x, \xi) W_{\hbar, E_N(\hbar)}(x, \xi) dx d\xi.$$
Since
$$ | \mathrm{Tr} Op_h^w(a) \Pi_{\hbar, E_N(\hbar)} | \leq ||Op^w(a)||_{L^2 \to L^2}\;  \mathrm{Tr}\Pi_{\hbar, E_N(\hbar)}, $$ 
 it follows that  the
 normalized  traces
 form a bounded family of linear functionals on the space $\Psi^0(\R^d)$  of zeroth order Weyl
 pseudo-differential operators. We refer to \cite[Section 24]{Sh} for background on norms of operators and norms of symbols (e.g. \cite[Problem 24.8]{Sh}).

   Assume that $E_N(\hbar) \to E$.  It follows from the exponential decay bounds that   $\{\frac{1}{\dim V_{\hbar, E_N(\hbar)}} W_{\hbar, E_N(\hbar)} \}$ is a tight sequence of signed measures of bounded mass  (i.e. $L^1$-norm)  on $\R^{2n}$. 
    Let $\nu$ be any limit measure. 

It follows that
\begin{equation} \label{IPTR} \begin{array}{l} 
E_N(\hbar)  \int_{T^* \R^d} a(x, \xi)W_{\hbar, E_N(\hbar)}(x, \xi) dx d\xi\\ \\ 
= 
 \int_{T^* \R^d} a(x, \xi) \left( - \frac{\hbar^2}{4}  (\Delta_{\xi} + \Delta_x) + (||x||^2 + ||\xi||^2) \right)  W_{\hbar, E_N(\hbar)}(x, \xi) dx d\xi\\ \\
= 
\int_{T^* \R^d} a(x, \xi) \left( (||x||^2 + ||\xi||^2) \right)  W_{\hbar, E_N(\hbar)}(x, \xi) dx d\xi  \\ \\ +
 \int_{T^* \R^d}\left( [- \frac{\hbar^2}{4}  (\Delta_{\xi} + \Delta_x)]  a(x, \xi)  \right)  W_{\hbar, E_N(\hbar)}(x, \xi) dx d\x\\ \\
  =   \int_{T^* \R^d} a(x, \xi) \left( (||x||^2 + ||\xi||^2) \right)  W_{\hbar, E_N(\hbar)}(x, \xi) dx d\xi + O(\hbar^2), \end{array} \end{equation}
  hence,
  \begin{equation} \label{ID2} \frac{1}{\dim V_{\hbar, E_N(\hbar)}} \;  \int_{T^* \R^d} a(x, \xi) \left( (||x||^2 + ||\xi||^2 - E_N(\hbar)) \right)  W_{\hbar, E_N(\hbar)}(x, \xi) dx d\xi = O(\hbar^2). \end{equation}
  
  Let $\wt{a}(\rho) = \int_{\{H = \rho\}} a(\rho \omega ) d S(\omega)$ be
  the normalized spherical means of $a$ with respect to Lebesgue measure
  on $\R^{2d}$. Also, let $\wcal_{\hbar, E_N(\hbar)}(\rho) = W_{\hbar, E_N(\hbar)}(x, \xi)$ be the radial expression for the Wigner distribution.
  Then, \eqref{ID2} becomes,
    \begin{equation} \label{ID2b}  \frac{1}{\dim V_{\hbar, E_N(\hbar)}} \int_0^{\infty} \wt{a}(\rho) \left(\rho - E_N(\hbar)) \right)  \wcal_{\hbar, E_N(\hbar)}(\rho) \rho^{d-1} d\rho = O(\hbar^2). \end{equation}
    Taking the limit $\hbar \to 0, E_N(\hbar) \to E$ and using  \eqref{ID2}-\eqref{ID2b},
    the weak* limit measure satisfies
    $$\int_0^{\infty} \tilde{a}(\rho) (\rho - E) d\nu = 0. $$
    It follows that $d\nu$ is supported on $\Sigma_E$ and since it is 
    a measure it is a constant multiple of $\delta_{\Sigma_E}$. The constant
    is fixed by \eqref{TRACEP}.


\end{proof}



 

\section{Appendix}

\subsection{Appendix on the Airy function}
The Airy function is defined by,

$$Ai(z) = \frac{1}{2 \pi i} \int_L e^{v^3/3 - z v} dv, $$
where $L$ is any contour that beings at a point at infinity in the sector $- \pi/2 \leq \arg (v) \leq - \pi/6$ and ends
at infinity in the sector $\pi/6 \leq \arg(v) \leq \pi/2$.  In the region $|\arg z| \leq (1 - \delta) \pi$
in $\C - \{\R_-\}$ write $v = z^{\half} + i t ^{\half}$ on the upper half of L and $v = z^{\half} - i t^{\half}$
in the lower half. Then
\begin{equation} \label{AIRYASYM} \Ai(z) = \Psi(z) e^{- \frac{2}{3} z^{3/2}}, \;
\mathrm{
with}\;
\Psi(z) \sim z^{-1/4} \sum_{j = 0}^{\infty} a_j z^{- 3j/2}, \;\; a_0 = \frac{1}{4} \pi^{-3/2}. \end{equation}

\subsection{Appendix on Laguerre functions \label{S:Laguerre}}

The Laguerre polynomials $L_k^{\alpha}(x)$  of degree $k$ and of type $\alpha$ on $[0, \infty)$  are defined by
\begin{equation}
e^{-x} x^{\alpha} L_k^{\alpha}(x) = \frac{1}{k!} \frac{d^k}{dx^k} (e^{-x}x^{k +\alpha}). \label{E:LagDef}
\end{equation}
They are solutions of the Laguerre equation(s),
$$x y'' + (\alpha + 1 -x) y(x)' + k y(x) = 0.$$

For fixed $\alpha$ they are orthogonal polyomials of $L^2(\R_+, e^{-x} x^{\alpha} dx) $. 
An othonormal basis is given by
$$\lcal_k^{\alpha}(x) = \left(\frac{\Gamma(k +1)}{\Gamma(k + \alpha+ 1)}\right)^{\half} L_k^{\alpha}(x).$$
We will have occasion to use the following generating function:
\[\sum_{k =0}^{\infty}  L_k^{\alpha}(x)  w^k = (1 -w)^{-\alpha -1} e^{- \frac{w}{1-w} x}\]

\noindent The most useful integral representation for the Laguerre functions is
\begin{equation}\label{INTnew}
e^{-x/2} L_n^{(\alpha)} (x) = \lr{-1}^n\oint \frac{e^{-\frac{x}{2}\cdot \frac{1-z}{1+z}}}{z^n\lr{1+z}^{\alpha+1}}\frac{dz}{2 \pi i z} ,
\end{equation}
where the contour encircles the origin once counterclockwise. Equivalently,
\begin{equation}\label{INT} e^{-x/2} L_n^{(\alpha)} (x) = \frac{(-1)^n}{2^{\alpha}}
\frac{1}{2 \pi i} \int^{1+} e^{- x z/2} \left( \frac{1 + z}{1 - z} \right)^{\nu/4} (1 - z^2)^{
\frac{\alpha-1}{2}} d z \end{equation}
where $\nu= 4n + \alpha + 2$ and the contour encircles $z = 1$ in the positive
direction and closes at $\Re z = \infty, |\Im z| = \mathrm{constant}$. In (5.9) of \cite{FW} the Laguerre functions are represented as the oscillatory integrals,
\begin{equation} \label{nu2}e^{- \nu t/2} L_n^{\alpha}(\nu t) = \frac{(-1)^n}{2^{\alpha}} \frac{1}{2 \pi i}
\int_{\lcal} [1 - z^2(u)]^{\frac{\alpha-1}{2}} \exp \{ \nu \left( \frac{u^3}{3} - B^2(t) u \right) \} du, \end{equation}
where $\nu = 4 n + 2 \alpha + 2$ and $B(t) $ is given by \eqref{BFORMULA}  (cf (5.5) of \cite{FW}) and and  $\lcal$ is a branch of the hyperbolic curve in the right half plane.

\subsection{Stationary Phase Expansion}\label{SO} We recall here the following simple version of the  stationary expansion, which we use in several proofs.
\begin{Lem}[\cite{Hor} Theorem 7.7.5]\label{L:SP LO}
Suppose $a,S\in \mathcal S(\R)$ and $S$ is a complex-valued phase function such that $\Im S|_{\supp(a)} \geq 0$ with a unique non-degenerate critical point at $t_0\in \supp(a)$
satisfying $\Im S(t_0)=0.$ Then, 
\begin{align}
\label{E:Stationary Phase 2}  I(h)=\int_{\R}e^{iS(x)/h}a(x)dx = 
 e^{i \frac{\pi}{4} \text{sgn} S''(0)} \lr{\frac{2\pi h}{\abs{S''(t_0)}}}^{1/2}\left[a(t_0)+O(h)\right],\qquad
\end{align}
\end{Lem}

\end{document}